\documentclass[journal,onecolumn]{IEEEtran}
\usepackage[latin1]{inputenc}
\usepackage[english]{babel}
\usepackage[margin=1in]{geometry}

\usepackage{enumitem}
\usepackage{blindtext}
\usepackage{amssymb}
\usepackage{hyperref}
\hypersetup{
    colorlinks=true,
    linkcolor=blue,
    filecolor=magenta,
    urlcolor=cyan,
}

\usepackage{amsthm}
\usepackage{amsmath}
\usepackage{booktabs}
\usepackage{array}
\usepackage{bm}
\newtheorem{theorem}{Theorem}[section]
\newtheorem{corollary}{Corollary}[section]
\newtheorem{lemma}[theorem]{Lemma}

\newtheorem{remark}{Remark}[section]
\newtheorem{example}{Example}[section]

\theoremstyle{definition}
\newtheorem{definition}{Definition}[section]

\title{Locally repairable convertible codes with optimal access costs}

\date{} 

\author{Xiangliang Kong
\thanks{Xiangliang Kong is with the Department of Electrical Engineering-Systems, Tel Aviv University, Tel Aviv-Yafo 6997801, Israel ({rongxlkong@gmail.com}). This work was supported by the European Research
Council (ERC) under Grant 852953.}}

\begin{document}

\maketitle
\begin{abstract}
Modern large-scale distributed storage systems use erasure codes to protect against node failures with low storage overhead. In practice, the failure rate and other factors of storage devices in the system may vary significantly over time, and leads to changes of the ideal code parameters. To maintain the storage efficiency, this requires the system to adjust parameters of the currently used codes. The changing process of code parameters on encoded data is called code conversion. 

As an important class of storage codes, locally repairable codes (LRCs) can repair any codeword symbol using a small number of other symbols. This feature makes LRCs highly efficient for addressing single node failures in the storage systems. In this paper, we investigate the code conversions for locally repairable codes in the merge regime. We establish a lower bound on the access cost of code conversion for general LRCs and propose a general construction of LRCs that can perform code conversions with access cost matching this bound. This construction provides a family of LRCs together with optimal conversion process over the field of size linear in the code length. 
\end{abstract}

\section{Introduction}

Large-scale cluster storage systems employ erasure codes to protect against node failures with low storage overhead \cite{Huang12,SAPDVCB13,Apache}. In such scenarios, a set of $k$ data symbols
is encoded as a codeword of length $n$ using an $[n,k]$ code and the $n$ codeword symbols are distributed across $n$ different storage nodes. The ratio $\frac{k}{n}$, referred to as the code rate, is chosen according to the failure rate of storage devices. In practice, the failure rate of storage devices in large-scale storage systems vary significantly over time. However, the rate of the erasure code employed in the system usually remains fixed. Consequently, as time goes by, storage systems utilizing erasure codes with static rates will either be overly resource-consuming, overly risky, or a mix of the two. To address this problem, \cite{KRG19} shows that changing the code rate over time in response to the variations of failure rates of storage devices yields substantial savings in storage space and operating costs. In this approach, one need to convert already-encoded data from an $[n^I,k^I]$ code to its encoding under another code with different parameters $[n^F,k^F]$. This process is known as \emph{code conversion} \cite{MR22a}, the code with parameters $[n^I,k^I]$ is called the \emph{initial code} and the code with parameters $[n^F,k^F]$ is called the \emph{final code}.

The default approach for code conversion is to re-encode the data from the initial code according to the finial code. However, such an approach requires accessing a large number of symbols, decoding the entire data, transferring the data over the network and then re-encoding it. This causes substantial I/O load and adversely affects the efficiency of the system \cite{KMSYRG20}. In the pursuit of designing codes that allow for efficient code conversions, Maturana and Rashmi introduced \emph{convertible codes} \cite{MR22a}. Their work, along with subsequent studies\cite{MMR20,MR23a,MR23b}, focused on conversions of MDS codes and MDS array codes, and obtained several fundamental results in the theoretical study of code conversions for storage codes.

As another important storage codes, locally repairable codes (LRCs) can repair any codeword symbol using only a small number of other symbols. More formally, an $(n,k,r)$ LRC is a code that produces an $n$-symbol codeword from $k$ information symbols, such that for any symbol within the codeword, one can use at most $r$ other symbols to recover it. This feature makes LRCs notably efficient for large-scale storage systems, particularly in scenarios involving single node failures. Due to their applications in distributed and cloud storage systems, LRCs have become one of the rapidly developing topics in coding theory over the past decade. The study of LRCs encompasses a broad range of topics, including constructions (e.g., \cite{HCL07,TB14,BT2017,LMX19,JMX20,CFMsT22}), bounds (e.g., \cite{GHSY12,TBF16,ABHMT18,MG19,WZL19,Roth22}), and diverse variations (e.g., \cite{PKLK12,WZ14,GXY19,CMST20,CMST20b,KWG21,XC22}).

In this paper, we focus on code conversion problem for LRCs in the merge regime. The merge regime corresponds to conversions where multiple codewords are merged into a single codeword. Specifically, we establish a lower bound on the number of accessed symbols during the conversion process for general LRCs in the merge regime, and we propose a general construction of LRCs that can perform code conversions with access cost matching this bound. It's noteworthy that during the preparation of this paper, Maturana and Rashmi \cite{MR23b} studied conversion problem for LRCs with information locality. Based on piggybacking frameworks, they proposed a construction technique for designing LRCs that can perform code conversion at a lower cost than the default approach. In contrast to their work, we focus on LRCs with all-symbol locality. Moreover, our construction is from the algebraic coding perspective, and it explicitly provides a family of LRCs together with optimal conversion process over the field of size linear in the code length.

\subsection{The Setup and Related Results}

Before formally defining and discussing the code conversion problem for LRCs, we begin with the basic definition and bounds for LRCs.

For positive integers $n$ and $1\leq a<b\leq n$, we denote $[n]$ as the set of integers $\{1,2,\ldots,n\}$ and $[a,b]$ as $\{a,a+1,\ldots,b\}$. Let $q$ be a prime power, we denote $\mathbb{F}_q$ as the finite field of $q$ elements. We say that $\mathcal{C}$ is an $(n,k)_q$ code if $\mathcal{C}\subseteq\mathbb{F}_q^n$ and $|\mathcal{C}|=q^k$. For a vector $\mathbf{v}\in \mathbb{F}_q^n$ and a subset $S\subseteq [n]$, we denote $\mathbf{v}|_{S}$ as the vector obtained by removing the coordinates outside $S$. Let $\mathcal{C}$ be an $(n,k)_q$ code and $S$ be a subset of $[n]$, the restriction of $\mathcal{C}$ to $S$ is defined as $\mathcal{C}|_{S}=\{\mathbf{c}|_{S}: \mathbf{c}\in \mathcal{C}\}$.

\begin{definition}(Locally repairable codes, \cite{GHSY12,PD14})\label{LRCs}
Let $r\leq k\leq n$ be positive integers. We say that a code $\mathcal{C}$ with parameters $(n,k)$ over alphabet $\Sigma$ has locality $r$ if for every codeword $\mathbf{c}\in \mathcal{C}$ and every $i\in [n]$, there exists a subset of coordinates $I_i\subset[n]\setminus \{i\}$ with $|I_i|\leq r$ such that $\mathbf{c}(i)$ can be recovered by symbols from $\mathbf{c}|_{I_i}$. The subset $I_i$ is called a recovering set for $i$ and the $\mathcal{C}$ is called an \emph{$(n,k,r)$-LRC}.
\end{definition}

As a natural generalization of the classical Singleton bound, \cite{GHSY12} and \cite{PD14} proved a bound on the minimum distance of an $(n,k,r)$-LRC. Here, we refer to the version from \cite{TB14}.

\begin{theorem}\cite{TB14}\label{Singleton_bound}
Let $\mathcal{C}$ be an $(n,k,r)$-LRC of cardinality $|\Sigma|^k$ over alphabet $\Sigma$, then:
\begin{itemize}
    \item The rate of $\mathcal{C}$ satisfies $\frac{k}{n}\leq \frac{r}{r+1}$.
    \item The minimum distance of $\mathcal{C}$ satisfies $d\leq n-k-\lceil\frac{k}{r}\rceil+2$.
\end{itemize}
Moreover, an $(n,k,r)$-LRC that achieves the bound on the distance with equality is called an optimal LRC.
\end{theorem}

Next, we present the formal framework introduced in \cite{MR22a} for studying code conversions.

Let the \emph{initial code} $\mathcal{C}^{I}$ be an $(n^{I},k^{I})_q$ code and the \emph{final code} $\mathcal{C}^{F}$ be an $(n^{F},k^{F})_q$ code. Consider a message vector of length $M=lcm(k^{I},k^{F})$ over $\mathbb{F}_q$, denoted by $\mathbf{m}\in \mathbb{F}_q^{M}$. Clearly, to encode $\mathbf{m}$, we need $\lambda^{I}=M/k^{I}$ codewords in $\mathcal{C}^{I}$ and $\lambda^{F}=M/k^{F}$ codewords in $\mathcal{C}^{F}$. To specify the information encoded in the codewords involved in conversion process, we define partitions $\mathcal{P}^I$ and $\mathcal{P}^{F}$ of $[M]$ as: $\mathcal{P}^{I}=\{P_1^{I},\ldots,P_{\lambda_{I}}^{I}\}$ consists of $\lambda_{I}$ disjoint subsets, each of size $k^I$ and $\mathcal{P}^{F}=\{P_1^{F},\ldots,P_{\lambda_{F}}^{F}\}$ consists of $\lambda_{F}$ disjoint subsets, each of size $k^F$. Then, for every $i\in [\lambda^I]$, denote $\mathbf{c}_{i}=\mathcal{C}^{I}(\mathbf{m}|_{P_i^I})$ as the \emph{initial codeword} in $\mathcal{C}^{I}$ that encodes the message in $\mathbf{m}$ indexed by $P_i^{I}$ and for every $j\in [\lambda^F]$, denote $\mathbf{d}_{j}=\mathcal{C}^{F}(\mathbf{m}|_{P_j^F})$ as the \emph{final codeword} in $\mathcal{C}^{F}$ that encodes the message in $\mathbf{m}$ indexed by $P_j^{F}$. Then, the conversion procedure from $\mathcal{C}^{I}$ to $\mathcal{C}^{F}$ is defined as follows.

\begin{definition}(Code Conversion, \cite{MR22a})\label{def1}
A \emph{conversion} from the initial code $\mathcal{C}^{I}$ to the final code $\mathcal{C}^{F}$ with initial partition $\mathcal{P}^I$ and the final partition $\mathcal{P}^{F}$ is a procedure, denoted by $T_{\mathcal{C}^{I}\rightarrow \mathcal{C}^{F}}$, that for any $\mathbf{m}\in \mathbb{F}_{q}^{M}$, takes the set of initial codewords $\{\mathbf{c}_{i}: 1\leq i\leq \lambda^{I}\}$ that encodes $\mathbf{m}$ as inputs, and output the set of final codewords $\{\mathbf{d}_{i}: 1\leq i\leq \lambda^{F}\}$ that encodes $\mathbf{m}$.
\end{definition}

\begin{definition}(Convertible Code, \cite{MR22a})\label{def2}
An $(n^{I},k^{I};n^{F},k^{F})$ convertible code over $\mathbb{F}_q$ is defined by:
\begin{itemize}
    \item [1.] A pair of codes $(\mathcal{C}^{I},\mathcal{C}^{F})$ where $\mathcal{C}^{I}$ is an $(n^I,k^I)_q$ code and $\mathcal{C}^{F}$ is an $(n^F,k^F)_q$ code.
    \item [2.] A pair of partitions $\mathcal{P}^I$ and $\mathcal{P}^{F}$ of $[M]$ ($M=lcm(k^{I},k^{F})$) such that each subset in $\mathcal{P}^I$ is of size $k^I$ and each subset in $\mathcal{P}^F$ is of size $k^F$.
    \item [3.] A conversion procedure $T_{\mathcal{C}^{I}\rightarrow \mathcal{C}^{F}}$ that on input $\{\mathbf{c}_{i}: 1\leq i\leq \lambda^{I}\}$ outputs $\{\mathbf{d}_{i}: 1\leq i\leq \lambda^{F}\}$.
\end{itemize}
\end{definition}

When both $\mathcal{C}^{I}$ and $\mathcal{C}^{F}$ are MDS codes, the corresponding convertible code $(\mathcal{C}^{I},\mathcal{C}^{F})$ is called an MDS convertible code. When $\mathcal{C}^{I}$ is an LRC with parameters $(n^I,k^I,r^I)$ and $\mathcal{C}^{F}$ is an LRC with parameters $(n^F,k^F,r^F)$, the corresponding convertible code $(\mathcal{C}^{I},\mathcal{C}^{F})$ is called an locally repairable convertible code (LRCC) with parameters $(n^{I},k^{I},r^{I};n^{F},k^{F},r^{F})$.

Next, we define the access cost of code conversion, which measures the number of symbols that are affected by the conversion procedure.

\begin{definition}(Access Cost, \cite{MR22a})\label{def3}
The \emph{read access cost} and \emph{write access cost} are defined as the total number of symbols read 
and written during the conversion procedure, respectively. The \emph{access cost} of a conversion
procedure is the sum of its read and write access costs. The access cost of a convertible code is the access cost of its conversion procedure.
\end{definition}

To specify the transfer of codeword symbols during the conversion, we classify them into the following three categories:
\begin{itemize}
    \item [(1)] \emph{remaining symbols}, which refers to symbols in the initial codewords that remain as is in the final codewords;
    \item [(2)] \emph{accessed symbols}, which refers to symbols of the initial codewords that are accessed (read);
    \item [(3)] \emph{new symbols}, which refers to symbols in the final codewords which are not remaining from initial codewords.
\end{itemize}

As two fundamental cases of code conversions, \emph{Merge regime} and \emph{split regime} were introduced in \cite{MR22a} and \cite{MR22b}, separately. In the merge regime, $k^F=\zeta k^I$ for some integer $\zeta\geq 2$. As consequence, the length of the message vector $\mathbf{m}$ is $M=lcm(k^{I},k^{F})=k^F$. Thus, the number of initial codewords is $\lambda^I=\zeta$ and the number of final codewords is $\lambda^F=1$. In the split regime, $k^I=\zeta k^F$. Thus, the number of initial codewords is $\lambda^I=1$ and the number of final codewords is $\lambda^F=\zeta$. As shown in \cite{MR22a}, in the merge regime, all pairs of initial and final partitions $(\mathcal{P}^{I},\mathcal{P}^{F})$ are equivalent. Therefore, in the following context, we always assume that $\mathcal{P}^{I}=\{[k^I],[k^I+1,2k^I],\ldots,[(\zeta-1)k^I+1,\zeta k^I]\}$ and $\mathcal{P}^{F}=\{[\zeta k^I]\}$ for all $(n^I,k^I;n^F,k^F=\zeta k^I)$ convertible codes.

Maturana and Rashmi \cite{MR22a} proved the following lower bound on the access cost of linear MDS convertible codes in the merge regime. They also provided two constructions and showed that this lower bound is optimal. Their construction for general parameters requires the field size $q$ to be at least $\Omega(2^{\zeta\cdot (n^I)^3})$. Conditioning on the existence of certain superregular Hankel arraires, they further reduce this requirement to $q\geq \max\{n^I,n^F\}$ when $n^F-k^F\leq \lfloor\frac{n^I-k^I}{\zeta}\rfloor$.
\begin{theorem}\label{thm1}
For all linear $(n^I,k;n^F,\zeta k)$ MDS convertible codes, the read access cost of conversion is at least $\zeta\cdot \min\{k,n^F-\zeta k\}$ and the write access cost is at least $n^F-\zeta k$. Further, if $n^I-k<n^F-\zeta k$, the read access cost is at least $\zeta k$.
\end{theorem}





\subsection{Our Contribution}

In this paper, we study the code conversion problem of locally repairable codes in the merge regime. Due to their good error correcting capabilities and strong connections with MDS codes, we limit our focus on optimal LRCs, that is, LRCs achieving the bound on the distance in Theorem \ref{Singleton_bound}. In the following context, the term LRCs will always refer to optimal LRCs unless otherwise specified.

Below, we summarize the main contributions of this paper.
\begin{itemize}
    \item First, by establishing proper maps from polynomials of degree at most $k-1$ to polynomials of degree at most $\zeta k-1$, we proprose a construction of MDS convertible codes with optimal access costs for general parameters using GRS codes. As a consequence, we obtain a family of $(n^I,k;n^F,\zeta k)$ MDS convertible codes with optimal access costs over finite field of size linear in $n^I$. 
    \item Second, based on the construction of optimal LRCs through ``good polynomials'' in \cite{TB14}, we extend the construction of MDS convertible codes and propose a construction of LRCCs. Similarly, a family of $(n^I,k,r;n^F,\zeta k,r)$ LRCCs over finite field of size linear in $n^I$ is obtained.
    \item Third, we prove a general lower bound on the access cost of an $(n^I,k;n^F,\zeta k)$ convertible code with the condition that the finial code is an LRC. As a consequence, this shows that the LRCCs obtained in Section \ref{sec_LRCC} admit optimal access costs.
\end{itemize}

\subsection{Notations}

Throughout the paper, we use bold lowercase letters to represent vectors and bold uppercase letters to represent matrices. We use $\mathbf{I}_r$ to denote the identity matrix of order $r$ and the subscript ``$r$'' is usually omitted if it's clear from the context. We use $(a_{i,j})_{1\leq i\leq m \atop 1\leq j\leq n}$ to denote the $m\times n$ matrix with $a_{i,j}$ as its $(i,j)$-th entry and we denote $\mathbf{0}_{m\times n}$ as the $m\times n$ matrix with all entries equal to $0$. 
Given a vector $\mathbf{v}$ of length $n$ over $\mathbb{F}_q$, we use $\text{diag}(\mathbf{v})$ to denote the diagonal matrix with $\mathbf{v}(i)$ as its $i$-th diagonal element. Let $s,t$ be positive integers and $A=\{a_1,\ldots, a_s\}\subseteq \mathbb{F}_q$, we denote
$$\mathbf{A}^{(t)}=(a_j^{i-1})_{1\leq i\leq t \atop 1\leq j\leq s}$$
as the $t\times s$ Vandermonde matrix with $(i,j)$-th entry $a_{j}^{i-1}$ and call $\mathbf{A}^{(t)}$ the $t\times s$ Vandermonde matrix generated by $A$. 
Moreover, for $\alpha\in \mathbb{F}_q$, we use $\alpha A$ to denote the subset $\{\alpha a_1,\ldots, \alpha a_s\}\subseteq \mathbb{F}_q$.

For integer $k\geq 1$, let $g_1,\ldots,g_k$ be $\mathbb{F}_q$-linearly independent polynomials in $\mathbb{F}_q[x]$ and denote $\text{Span}_{\mathbb{F}_q}\{g_1,\ldots,g_k\}$ as the $k$-dim subspace of $\mathbb{F}_q[x]$ spanned by $g_1,\ldots,g_k$. Specially, when $g_i=x^{i-1}$, we denote $\mathbb{F}_q^{<k}[x]$ as $\text{Span}_{\mathbb{F}_q}\{g_1,\ldots,g_k\}$ for simplicity. Let $f\in \text{Span}_{\mathbb{F}_q}\{g_1,\ldots,g_k\}$ and assume that $f=\sum_{i=1}^{k}f_ig_i$, where $f_i\in \mathbb{F}_q$. We denote $\mathbf{v}_f=(f_1,\ldots,f_k)$ and call $\mathbf{v}_f$ the coefficient vector of $f$ w.r.t. basis $\{g_i\}_{i=1}^{k}$. When the basis $\{g_i\}_{i=1}^{k}=\{x^{i-1}\}_{i=1}^{k}$, $\mathbf{v}_f$ is called the coefficient vector of $f$ for brevity. Moreover, for a subset $A\subseteq \mathbb{F}_q$, we use $h_{A}(x)=\prod_{a\in A}(x-a)$ to denote the annihilator polynomial of $A$.

\subsection{Origanization of the Paper}

The rest of the paper is structured as follows. In Section \ref{sec_MDSC}, we provide a general construction of MDS convertible code with optimal access cost in the merge regime. In Section \ref{sec_LRCC}, we extended the construction of MDS convertible codes and provide a construction of LRCCs with optimal access costs in the merge regime. Then, in Section \ref{sec_lb}, we prove a general lower bound on the access cost of the conversion process for LRCCs in the merge regime. Finally, we conclude the paper by highlighting some open questions in Section \ref{sec_con}.

\section{Constructions of MDS convertible codes}\label{sec_MDSC}

In this section, we provide a general construction of MDS convertible codes with optimal access costs in the merge regime. We first provide an example to illustrate the underlying idea of our general construction. Then, in Section \ref{sec_MDSC2}, we present our construction of MDS convertible codes. Though this construction,  in Section \ref{sec_MDSC3}, we obtain a family of $(n^I,k;n^F,\zeta k)$ MDS convertible codes with optimal access costs over finite field of size linear in $n^I$. 

Recall that in the merge regime, we have $k^F=\zeta k^I$ for some integer $\zeta\geq 2$. Write $l^I=n^I-k^I$ and $l^F=n^F-\zeta k^I$. Note that when $l^I< l^F$ or $k^I\leq l^F$, the lower bound in Theorem \ref{thm1} can be obtained by the default approach. Thus, throughout this section, we assume that $l^F\leq \min\{l^I,k^I\}$.

\subsection{An illustrative example}\label{sec_MDSC1}

We start with the following example.

\begin{example}\label{ex1}
In this example, let $\mathbb{F}_q=\mathbb{F}_{19}$, we construct a $(6,4)_q$ MDS code $\mathcal{C}^I$, a $(10,8)_q$ MDS code $\mathcal{C}^F$ and a conversion procedure $T$ from $\mathcal{C}^I$ to $\mathcal{C}^F$ with write access cost $2$ and read access cost $4$. Together, they form a $(6,4;10,8)$ MDS convertible codes with optimal access cost. Note that $2$ is generator of $\mathbb{F}_{19}^{*}$. We list all the values of $2^x$ in Table I, which might help with calculations in this example.
\begin{table}[h!]
\centering
\caption{Correspondence of $2^x$ in $\mathbb{F}_{19}^{*}$}
\begin{tabular}{|c|c|c|c|c|c|c|c|c|c|}
\hline
$x$ & 0 & 1 & 2 & 3 & 4 & 5 & 6 & 7 & 8 \\
\hline
$2^x$ & 1 & 2 & 4 & 8 & 16 & 13 & 7 & 14 & 9 \\
\hline
$x$ & 9 & 10 & 11 & 12 & 13 & 14 & 15 & 16 & 17 \\
\hline
$2^x$ & 18 & 17 & 15 & 11 & 3 & 6 & 12 & 5 & 10 \\
\hline
\end{tabular}
\end{table}

Let $A=\{a_1,a_2,a_3,a_4\}=\{1,8,7,18\}$, $B=\{b_1,b_2,b_3,b_4\}=\{2,16,14,17\}$ and $C=\{c_1,c_2\}=\{4,9\}$. Then, it holds that $a_i=2^{3(i-1)}$, $b_i=2a_i$ and $c_i=4a_i$. We define the initial code $\mathcal{C}^I$ as the $[6,4]$ RS code with evaluation points $A\cup C$, i.e., 
$$\mathcal{C}^I=\{(f(\alpha), \alpha\in A\cup C):~f\in \mathbb{F}_q^{<4}[x]\}.$$

Next, to define the final code $\mathcal{C}^F$, we need some preliminaries. Let $\bm{\theta}=(11,3,3,6)$ and
$$\mathbf{M}=
\mathbf{A}^{(4)}\cdot \text{diag}(\bm{\theta})\cdot(\mathbf{B}^{(4)})^{-1}=
\left(\begin{array}{cccc}
0 & 16 & 14 & 7\\
1 & 17 & 17 & 17\\
16 & 16 & 10 & 7\\
1 & 6 & 17 & 15
\end{array}\right).
$$
Then, one can verify that the following holds:
\begin{itemize}
    \item [1.] For every $i\in [4]$, $\mathbf{M}\cdot (1,b_i,b_i^2,b_i^3)^{T}=\bm{\theta}(i) (1,a_i,a_i^2,a_i^3)^{T}$.
    \item [2.] $\mathbf{M}\cdot(1,c_1,c_1^2,c_1^3)^{T}=(14,4,4,3)^{T}=13 (1,c_1,c_1^2,c_1^3)^{T}+(1,c_2,c_2^2,c_2^3)^{T}$.
    \item [3.] $\mathbf{M}\cdot(1,c_2,c_2^2,c_2^3)^{T}=(16,16,12,17)^{T}=18 (1,c_1,c_1^2,c_1^3)^{T}-2 (1,c_2,c_2^2,c_2^3)^{T}$.
\end{itemize}
For each $f\in \mathbb{F}_q^{<4}[x]$, we denote $M(f)$ as the polynomial in $\mathbb{F}_q^{<4}[x]$ with coefficient vector $\mathbf{v}_f\cdot \mathbf{M}$. Then, we have
\begin{equation*}
    \begin{cases}
    M(f)(b_i)=\mathbf{v}_f\cdot \mathbf{M}\cdot (1,b_i,b_i^2,b_i^3)^{T}=\bm{\theta}(i)f(a_i),~1\leq i\leq 4;\\
    M(f)(c_1)=\mathbf{v}_f\cdot \mathbf{M}\cdot (1,c_1,c_1^2,c_1^3)^{T}=13f(c_1)+f(c_2);\\
    M(f)(c_2)=\mathbf{v}_f\cdot \mathbf{M}\cdot (1,c_2,c_2^2,c_2^3)^{T}=18f(c_1)-2f(c_2).
    \end{cases}
\end{equation*}

Recall that $h_A$ and $h_B$ are the annihilator polynomials of $A$ and $B$, respectively. Let $T$ be the map from $\mathbb{F}_q^{<4}[x]\times \mathbb{F}_q^{<4}[x]$ to $\mathbb{F}_q^{<8}[x]$ such that
$$T:~(f_1,f_2)\longmapsto h_B\cdot f_1+h_A\cdot M(f_2)$$
and let $\mathbf{u}=(u_\alpha, \alpha\in A\cup B\cup C)$ be a vector in $\mathbb{F}_{q}^{10}$ with coordinates indexed by elements in $A\cup B\cup C$ and 
\begin{equation*}
    \begin{cases}
    u_{a_i}=h_B(a_i)^{-1},~1\leq i\leq 4;\\
    u_{b_i}=\bm{\theta}(i)^{-1}h_A(b_i)^{-1},~1\leq i\leq 4;\\
    u_{c_i}=1,~1\leq i\leq 2.\\
    \end{cases}
\end{equation*}
Now, we define the final code $\mathcal{C}^{F}$ as
$$\mathcal{C}^F=\{(u_\alpha T(f_1,f_2)(\alpha), \alpha\in A\cup B\cup C):~f_1,f_2\in \mathbb{F}_q^{<4}[x]\}.$$
Since $T(f_1,f_2)\in \mathbb{F}_q^{<8}[x]$, $\mathcal{C}^F$ is a subcode of the $[10,8]$ GRS code defined by $A\cup B\cup C$ and vector $\mathbf{u}$. Note that
\begin{equation*}
    u_\alpha T(f_1,f_2)(\alpha)=\begin{cases}
    f_1(\alpha),~\text{when}~\alpha\in A;\\
    f_2(\alpha),~\text{when}~\alpha\in B.
    \end{cases}
\end{equation*}
Thus, we can use $(u_\alpha T(f_1,f_2)(\alpha),\alpha\in A)$ and $(u_\alpha T(f_1,f_2)(\alpha),\alpha\in B)$ to reconstruct $f_1$ and $f_2$, respectively. This implies that $\dim(\mathcal{C}^F)=8$ and $\mathcal{C}^F$ is a $[10,8]$ GRS code. Moreover, note that
\begin{equation*}
    \begin{cases}
    u_{c_1} T(f_1,f_2)(c_1)=h_B(c_1)f_1(c_1)+h_A(c_1)(13f_2(c_1)+f_2(c_2)),\\
    u_{c_2} T(f_1,f_2)(c_2)=h_B(c_2)f_1(c_2)-h_A(c_2)(f_2(c_1)+2f_2(c_2)).
    \end{cases}
\end{equation*}
Therefore, for any two codewords $\mathbf{c}_1=(f_1(\alpha), \alpha\in A\cup C)$ and $\mathbf{c}_2=(f_2(\alpha), \alpha\in A\cup C)$ in $\mathcal{C}^I$, the map $T$ induces a conversion procedure from $(\mathbf{c}_1,\mathbf{c}_2)$ to $\mathbf{d}=(u_\alpha T(f_1,f_2)(\alpha),\alpha\in A\cup B\cup C)\in \mathcal{C}^F$ such that 
\begin{itemize}
    \item [1)]  the $4$ symbols in $\mathbf{c}_1|_{A}$ remain in $\mathbf{d}$ as $\mathbf{d}|_{A}$, and the $4$ symbols in $\mathbf{c}_2|_{A}$ remain in $\mathbf{d}$ as $\mathbf{d}|_{B}$\footnote{Let $\mathcal{C}$ be an $[n,k]$ RS code with evaluation points $A$. For any $\mathbf{c}=(f(\alpha),\alpha\in A)\in \mathcal{C}$ and $A'\subseteq A$, we denote $\mathbf{c}|_{A'}$ as the restriction of $\mathbf{c}$ to coordinates corresponding to $A'$, i.e., $\mathbf{c}|_{A'}=(f(\alpha),\alpha\in A')$.};
    \item [2)] the $2$ symbols in $\mathbf{d}|_{C}$ are linear combinations of symbols in $\mathbf{c}_1|_C$ and $\mathbf{c}_2|_C$.
\end{itemize}
Therefore, the conversion procedure $T$ has write access cost $2$ and read access cost $4$.
\end{example}


\begin{remark}\label{rmk2-1}
Example \ref{ex1} provides an overview of the underlying idea behind our general construction. That is, by selecting proper evaluation points $A=\{a_1,\ldots,a_k\}$, $B=\{b_1,\ldots,b_k\}$ and $C=\{c_1,\ldots,c_l\}$, we can obtain a surjective map $T:\mathbb{F}_q^{<k}[x]\times \mathbb{F}_q^{<k}[x] \rightarrow \mathbb{F}_q^{<2k}[x]$ such that for every $(f_1,f_2)\in\mathbb{F}_q^{<k}[x]\times \mathbb{F}_q^{<k}[x]$ the following holds:
\begin{itemize}
    \item [1.] $T(f_1,f_2)(a_i)=\alpha_if_1(a_i)$ and $T(f_1,f_2)(b_i)=\beta_if_2(a_i)$, where for each $1\leq i\leq k$, $\alpha_i$ and $\beta_i$ are non-zero constants determined by $A$, $B$ and $C$;
    \item [2.] $T(f_1,f_2)(c_i)=\sum_{j=1}^{l}(\gamma_{i,j}^{(1)}f_1(c_j)+\gamma_{i,j}^{(2)}f_2(c_j))$, where for each $(i,j)\in [l]\times [l]$, $\gamma_{i,j}^{(1)}$ and $\gamma_{i,j}^{(2)}$ are constants determined by $A$, $B$ and $C$.
\end{itemize}
Then, based on this, we can define $\mathcal{C}^{I}$ as the $[k+l,k]$ RS code with evaluation points $A\cup C$ and $\mathcal{C}^F$ as
$$\mathcal{C}^F=\{(u_\alpha T(f_1,f_2)(\alpha),\alpha\in A\cup B\cup C):~f_1,f_2\in \mathbb{F}_q^{<k}[x]\},$$ 
where $u_{a_i}=\alpha_i^{-1}$, $u_{b_i}=\beta_i^{-1}$ and $u_{c_i}=1$. Since $T$ is surjective, $\mathcal{C}^F$ is actually the $[2k+l,k]$ GRS code defined by $A\cup B\cup C$ and vector $(u_{\alpha},\alpha\in A\cup B\cup C)$. 

By the one-to-one correspondence between polynomials and codewords in GRS codes, $T$ induces a conversion procedure from $\mathcal{C}^{I}$ to $\mathcal{C}^{F}$, i.e., codewords $\mathbf{c}_1=(f_1(\alpha),\alpha\in A\cup C)$ and $\mathbf{c}_2=(f_2(\alpha),\alpha\in A\cup C)$ are mapped to $\mathbf{d}=(u_\alpha T(f_1,f_2)(\alpha),\alpha\in A\cup B\cup C)$. Similar to Example \ref{ex1}, one can easily verify that $\mathbf{d}|_{A}=\mathbf{c}_1|_{A}$, $\mathbf{d}|_{B}=\mathbf{c}_2|_{A}$ and symbols in $\mathbf{d}|_{C}$ are linear combinations of symbols in $\mathbf{c}_1|_{C}$ and $\mathbf{c}_2|_{C}$. This will imply that $(\mathcal{C}^{I},\mathcal{C}^{F})$ is an MDS convertible codes with optimal access cost.
\end{remark}

\begin{remark}\label{rmk2-2}
In Example \ref{ex1}, the map $T$ is defined as 
$$T:~(f_1,f_2)\longmapsto h_B\cdot f_1+h_A\cdot M(f_2),$$
where $M(f_2)$ is the polynomial with coefficient vector $\mathbf{v}_{f_2}\cdot \mathbf{M}$. Note that $\mathbf{M}$ is a $k\times k$ matrix satisfying 
\begin{itemize}
    \item [1.] $\mathbf{M}\cdot (1,b_i,\ldots,b_i^{k-1})^{T}=\theta_i(1,a_i,\ldots,a_i^{k-1})^{T}$, where for each $1\leq i\leq k$, $\theta_i$ is non-zero constant determined by $A$, $B$ and $C$;
    \item [2.] $\mathbf{M}\cdot(1,c_i,\ldots,c_i^{k-1})^{T}\in \text{Span}_{\mathbb{F}_q}\{(1,c_i,\ldots,c_i^{k-1})^{T}:~1\leq i\leq l\}.$
\end{itemize}
One can easily check that these properties of $\mathbf{M}$ ensure the map $T$ defined above satisfies the requirements in Remark \ref{rmk2-1}. Thus, the crucial part of the construction is to find proper subsets $A$, $B$ and $C$ such that there exists a $k\times k$ matrix $\mathbf{M}$ satisfying the above requirements.
\end{remark}

\subsection{Constructions of MDS convertible codes with optimal access costs}\label{sec_MDSC2}

First, we assume that $l^F=l^I=l<k^I$ and present a construction of MDS convertible codes for this case. Then, we will modify it for the case when $l^F<l^I$.

\textbf{Construction I}: Let $A_i=\{a_{i,1},\ldots, a_{i,k}\}$, $1\leq i\leq \zeta$ and $C=\{c_1,\ldots, c_l\}$ be mutually disjoint subsets of $\mathbb{F}_q$ such that for every $2\leq i\leq \zeta$, there is a $k\times k$ matrix $\mathbf{M}_i$ satisfying:
\begin{description}
    \item [Condition 1] For every $1\leq j\leq k$, $$\mathbf{M}_i\cdot \mathbf{a}_{i,j}=\theta_{i,j} \mathbf{a}_{1,j},$$
    where $\mathbf{a}_{i,j}=(1,a_{i,j},\ldots,a_{i,j}^{k-1})^{T}$ and $\theta_{i,j}=\frac{h_{A_i\cup C\setminus\{a_{i,j}\}}(a_{i,j})}{h_{A_1\cup C\setminus\{a_{1,j}\}}(a_{1,j})}$.
    \item [Condition 2] For every $1\leq j\leq l$, $\mathbf{M}_i\cdot(1,c_j,\ldots,c_j^{k-1})^{T}\in \text{Span}_{\mathbb{F}_q}\{(1,c_j,\ldots,c_j^{k-1})^{T}:~1\leq j\leq l\}$.
\end{description}
Let $\mathcal{C}^I$ be the $[k+l,k]$ RS code with evaluation points $A_1\cup C$, i.e.,
\begin{equation*}
    \mathcal{C}^I=\{(f(\alpha),~\alpha\in A_1\cup C): f\in \mathbb{F}_q^{<k}[x]\}.
\end{equation*}
Denote $A=\bigcup_{i=1}^{k}A_i$. Let $\mathbf{M}_1=\mathbf{I}$ and define map $T:(\mathbb{F}_q^{<k}[x])^{\zeta}\rightarrow \mathbb{F}_q^{<\zeta k}[x]$
as
\begin{align}\label{cons1-1}
   T: (f_1,\ldots,f_{\zeta})&\longmapsto \sum_{i=1}^{\zeta}h_{A\setminus A_i}\cdot M_i(f_i),
\end{align}
where $M_i(f_i)\in \mathbb{F}_q^{<k}[x]$ is the polynomial with coefficient $\mathbf{v}_{f_i}\cdot \mathbf{M}_i$. Then, we define $\mathcal{C}^F$ as
\begin{align}\label{cons1-2}
  \mathcal{C}^F =\{(u_{\alpha}T(f_1,\ldots,f_{\zeta})(\alpha),~\alpha\in A\cup C): f_i\in \mathbb{F}_q^{<k}[x],~1\leq i\leq \zeta\},
\end{align}
where $u_{\alpha}=\theta_{i,j}^{-1}h_{A\setminus A_i}^{-1}(a_{i,j})$ when $\alpha=a_{i,j}\in A_i$ and $u_{\alpha}=1$ when $\alpha\in C$. 

\begin{theorem}\label{thm3-1}
For positive integers $\zeta\geq 2$, $k$ and $l<k$, the $(\mathcal{C}^I, \mathcal{C}^{F})$ given by Construction I is a $(k+l,k;\zeta k+l,\zeta k)$ MDS convertible code with optimal access cost.
\end{theorem}

\begin{proof}
Let $\mathbf{c}_1,\ldots,\mathbf{c}_{\zeta}\in \mathcal{C}^{I}$ be the $\zeta$ initial codewords and denote $f_i(x)$ as the encoding polynomial of $\mathbf{c}_i$. Since for every $1\leq i\leq \zeta$, $\deg(M_i(f_i))\leq k-1$ and $\deg(h_{A\setminus A_i})=(\zeta-1)k$, it holds that $T(f_1,\ldots,f_{\zeta})\in \mathbb{F}_q^{<\zeta k}$. Thus, $\mathcal{C}^{F}$ is a subcode of the $[\zeta k+l,\zeta k]$ GRS code defined by evaluation points $A\cup C$ and vector $(u_\alpha, \alpha\in A\cup C)$. Moreover, the map $T$ induces a conversion procedure from $\mathcal{C}^{I}$ to $\mathcal{C}^{F}$, that is, $\mathbf{c}_1,\ldots,\mathbf{c}_{\zeta}$ are converted to the codeword 
$$\mathbf{d}=(u_{\alpha}T(f_1,\ldots,f_{\zeta})(\alpha),~\alpha\in A\cup C)\in \mathcal{C}^{F}.$$


Next, we show that the read and write access costs of the conversion procedure induced by $T$ are $\zeta l$ and $l$, respectively. By Theorem \ref{thm1}, this implies that $(\mathcal{C}^I,\mathcal{C}^{F})$ has optimal access cost. As a byproduct, we shall see that $\mathcal{C}^{F}$ has dimension $\zeta k$, i.e., the map $T$ is surjective. This yields the result.


For each $(i,j)\in [\zeta]\times [k]$ and $\alpha=a_{i,j}$, since $h_{A\setminus A_i}$ is the annihilator polynomial of $A\setminus A_i$, we have 
\begin{align*}
    u_{\alpha}T(f_1,\ldots,f_{\zeta})(\alpha)&=\theta_{i,j}^{-1}h_{A\setminus A_i}^{-1}(a_{i,j})\sum_{s=1}^{\zeta}h_{A\setminus A_s}(a_{i,j})(M_s(f_s))(a_{i,j})\\
    &=\theta_{i,j}^{-1}(M_i(f_i))(a_{i,j})\\
    &=f_i(a_{1,j}),
\end{align*}
where the last equality follows from $(M_i(f_i))(a_{i,j})=\mathbf{v}_{f_i}\cdot \mathbf{M}_i\cdot \mathbf{a}_{i,j}=\theta_{i,j}f_i(a_{1,j})$.
Therefore, we have $\mathbf{d}|_{A_i}=\mathbf{c}_i|_{A_1}$. In other words, there are at least $k$ symbols from each $\mathbf{c}_i$ remaining in $\mathbf{d}$. Thus, the write access cost is at most $l$. Moreover, note that $\mathcal{C}^{I}|_{A_1}=\mathbb{F}_q^{k}$. Thus, we have $\mathcal{C}^{F}|_{A}=\mathbb{F}_q^{\zeta k}$, which leads to $\dim(\mathcal{C}^{F})=\zeta k$.

On the other hand, for $1\leq j\leq l$, denote $\mathbf{c}_j=(1,c_j,\ldots,c_j^{k-1})^{T}$. By $\mathbf{M}_i\cdot \mathbf{c}_j\in \text{Span}_{\mathbb{F}_q}\{\mathbf{c}_1,\ldots,\mathbf{c}_l\}$, we can assume that $\mathbf{M}_i\cdot \mathbf{c}_j=\sum_{s=1}^{l}\eta_{s}^{(i,j)}\mathbf{c}_{s}$. Then, for $\alpha=c_j\in C$, we have
\begin{align}
    u_{\alpha}T(f_1,\ldots,f_{\zeta})(\alpha)
    &=\sum_{i=1}^{\zeta}h_{A\setminus A_i}(c_j)(M_i(f_i))(c_j)\nonumber\\
    &=\sum_{i=1}^{\zeta}h_{A\setminus A_i}(c_j)\left(\mathbf{v}_{f_i}\cdot\sum_{s=1}^{l}\eta_{s}^{(i,j)}\mathbf{c}_{s}\right)\nonumber\\
    &=\sum_{i=1}^{\zeta}\sum_{s=1}^{l}h_{A\setminus A_i}(c_j)\eta_{s}(i,j)f_{i}(c_{s}).
    \label{eq3-3-1}
\end{align}
This implies that the $l$ symbols in $\mathbf{d}|_{C}$ are linear combinations of the $\zeta l$ symbols in $\mathbf{c}_1|_{C},\ldots,\mathbf{c}_\zeta|_{C}$. Therefore, the read access cost is at most $\zeta l$. 
\end{proof}

Now, we present the construction for the general case when $l^F\leq l^I$. The idea is simple. Note that if $\mathcal{C}^{I}$ is a linear $[k+l^I,k]$ MDS code, then, $\mathcal{C}^{I}|_{[k+l^F]}$ is a $[k+l^F,k]$ MDS code. Therefore, we can use the same $\mathcal{C}^{F}$ obtained in Construction I as the final code and maintain the access cost.

\textbf{Construction II}: Let $A_1,\ldots,A_{\zeta}$, $C$ and $(u_{\alpha},\alpha\in\bigcup_{i\in [\zeta]}{A_i}\cup C)$ be identical to those defined in Construction I. Let $B=\{b_{1},\ldots, b_{l^I-l^F}\}$ be a subset of $\mathbb{F}_q$ of size $(l^I-l^F)$ such that $B\cap (A_1\cup C)=\emptyset$\footnote{When $l^I=l^F$, $B=\emptyset$}. Let $\mathcal{C}^I$ be the $[k+l^I,k]$ RS code with evaluation points $A_1\cup C\cup B$, i.e.,
\begin{equation*}
    \mathcal{C}^I=\{(f(\alpha),~\alpha\in A_1\cup C\cup B): f(x)\in \mathbb{F}_q^{<k}[x]\}.
\end{equation*}
Let the conversion procedure $T$ and the final code $\mathcal{C}^F$ be the same as those defined in (\ref{cons1-1}) and (\ref{cons1-2}).

By Theorem \ref{thm1} and Theorem \ref{thm3-1}, we have the following immediate result.

\begin{theorem}\label{thm2}
For positive integers $\zeta\geq 2$, $k$, $l^I$ and $l^F\leq \min\{k,l^I\}$, the $(\mathcal{C}^I, \mathcal{C}^{F})$ given by Construction II is a $(k+l^I,k;\zeta k+l^F,\zeta k)$ MDS convertible code with optimal access cost.
\end{theorem}

\subsection{A family of MDS convertible codes with optimal access costs}\label{sec_MDSC3}

From Section \ref{sec_MDSC2}, the key to implementing Construction I and II lies in finding pairwise disjoint subsets $A_i$'s and $C$ from $\mathbb{F}_q$, such that there exist matrices $\mathbf{M}_i$'s satisfying conditions 1 and 2 in Construction I. In the following, we first give a sufficient condition for when $A_i$'s can ensure the existence of such $\mathbf{M}_i$'s. Then, we give an explicit construction of $A_1$, $\ldots$, $A_{\zeta}$ satisfying this sufficient condition. Together with proper subsets $B$ and $C$, we provide an implementation of Construction II, and obtain a family of MDS convertible codes with optimal access costs.

\begin{lemma}\label{lem2-1}
For positive integers $k\geq 2$ and $l<k$, let $A=\{a_{1},\ldots,a_{k}\}$ and $B=\{b_{1},\ldots,b_{k}\}$ be disjoint subsets of $\mathbb{F}_q$. Suppose that there is an invertible matrix $\mathbf{T}$ over $\mathbb{F}_q$ of order $l$ satisfying
\begin{equation}\label{eq3-1-1}
\mathbf{B}^{(l)}=\mathbf{T}\mathbf{A}^{(l)}.
\end{equation}
Then, for any $C=\{c_1,\ldots,c_l\}\subseteq \mathbb{F}_q$, there is a $k\times k$ matrix $\mathbf{M}$ such that:
\begin{itemize}
    \item [1.] For every $1\leq i\leq k$, $\mathbf{M}\cdot (1,b_i,\ldots,b_i^{k-1})^{T}=\theta_i (1,a_i,\ldots,a_i^{k-1})^{T}$,
    where $\theta_i=\frac{h_{B\cup C\setminus\{b_i\}}(b_i)}{h_{A\cup C\setminus\{a_i\}}(a_i)}$. 
    \item [2.] For every $1\leq i\leq l$, $\mathbf{M}\cdot(1,c_i,\ldots,c_i^{k-1})^{T}\in \text{Span}_{\mathbb{F}_q}\{(1,c_i,\ldots,c_i^{k-1})^{T}:~1\leq i\leq l\}$.
\end{itemize}
\end{lemma}

By Lemma \ref{lem2-1}, we have the following immediate corollary, which provides a sufficient condition for the existence of $\mathbf{M}_i$'s satisfying conditions 1 and 2 in Construction I.

\begin{corollary}\label{coro2-1}
Let $\zeta \geq 2$, $k$ and $l<k$ be positive integers. Let $A_i=\{a_{i,1},\ldots,a_{i,k}\}$, $1\leq i\leq \zeta$, and $C=\{c_1,\ldots, c_l\}$ be mutually disjoint subsets of $\mathbb{F}_q$. If for every $2\leq i\leq \zeta$, there is an invertible matrix $\mathbf{T}_i$ over $\mathbb{F}_q$ of order $l$ satisfying
\begin{equation*}
\mathbf{A}_i^{(l)}=\mathbf{T}_i\mathbf{A}_1^{(l)}.
\end{equation*}
Then, for every $2\leq i\leq \zeta$, there is a $k\times k$ matrix $\mathbf{M}_i$ satisfies condition 1 and 2 in Construction I.
\end{corollary}

Now, we present the proof of Lemma \ref{lem2-1}.

\begin{proof}[Proof of Lemma \ref{lem2-1}]
Define $\mathbf{M}=\mathbf{A}^{(k)}\cdot \text{diag}(\theta_1,\ldots,\theta_k)\cdot (\mathbf{B}^{(k)})^{-1}$. 
Let $\mathbf{a}_i=(1,a_i,\ldots,a_i^{k-1})^{T}$, $\mathbf{b}_i=(1,b_i,\ldots,b_i^{k-1})^{T}$, $1\leq i\leq k$ and $\mathbf{c}_i=(1,c_i,\ldots,c_i^{k-1})^{T}$, $1\leq i\leq l$. Then, for every $1\leq i\leq k$, $\mathbf{M}\cdot \mathbf{b}_i=\theta_i \mathbf{a}_i$. Thus, we only need to verify that $\mathbf{M}$ satisfies the second condition.

For every $1\leq j\leq l$, it holds that
\begin{equation*}\label{eq3-1-2}
    \mathbf{c}_j=\sum_{i=1}^{k}\frac{h_{B\setminus\{b_i\}}(c_j)}{h_{B\setminus\{b_i\}}(b_i)}\mathbf{b}_i.
\end{equation*}
Denote $\eta_{i,j}=\frac{h_{B\setminus\{b_i\}}(c_j)}{h_{B\setminus\{b_i\}}(b_i)}$.
This leads to $\mathbf{M}\cdot \mathbf{c}_j=\sum_{i=1}^{k}\eta_{i,j}\theta_i\mathbf{a}_i.$

For each $f(x)\in \mathbb{F}_q^{<k-l}[x]$, we denote $\mathbf{v}_{h_C\cdot f}$ as the coefficient vector of polynomial $h_C(x)f(x)$. Then,
\begin{align}
    \mathbf{v}_{h_C\cdot f}\cdot \mathbf{M}\cdot \mathbf{c}_j&=\sum_{i=1}^{k}\eta_{i,j}\theta_i(\mathbf{v}_{h_C\cdot f}\cdot \mathbf{a}_i)= \sum_{i=1}^{k}\eta_{i,j}\theta_i h_C(a_i)f(a_i) \nonumber\\
    &=\sum_{i=1}^{k}\frac{h_{B\setminus\{b_i\}}(c_j)}{h_{B\setminus\{b_i\}}(b_i)}
    \frac{h_{B\cup C\setminus\{b_i\}}(b_i)}{h_{A\cup C\setminus\{a_i\}}(a_i)} 
    h_C(a_i)f(a_i)\nonumber\\
    &=\sum_{i=1}^{k}\frac{h_{C}(b_i)}{h_{A\setminus\{a_i\}}(a_i)} h_{B\setminus\{b_i\}}(c_j)f(a_i). \label{eq3-1-3}
\end{align}
Note that $h_{C}(b_i)h_{B\setminus\{b_i\}}(c_j)=\prod_{c\in C}(b_i-c)\prod_{b\in B\setminus\{b_i\}}(c_j-b)=-h_{C\setminus\{c_j\}}(b_i)h_{B}(c_j)$. Thus, (\ref{eq3-1-3}) is equivalent to 
\begin{align}
    \mathbf{v}_{h_C\cdot f}\cdot \mathbf{M}\cdot \mathbf{c}_j
    &=-h_{B}(c_j)\left(\sum_{i=1}^{k}\frac{h_{C\setminus\{c_j\}}(b_i)f(a_i)}{h_{A\setminus\{a_i\}}(a_i)}\right). \label{eq3-1-4}
\end{align}
Denote $\mathbf{h}_{C\setminus\{c_j\}}$ as the coefficient vector of polynomial $h_{C\setminus\{c_j\}}(x)$. By (\ref{eq3-1-1}), we have 
$$\mathbf{h}_{C\setminus\{c_j\}}\cdot\mathbf{B}^{(l)}=\mathbf{h}_{C\setminus\{c_j\}}\cdot\mathbf{T}\mathbf{A}^{(l)}.$$
Denote $H(x)\in \mathbb{F}_q^{<l}[x]$ as the polynomial with coefficient vector $\mathbf{h}_{C\setminus\{c_j\}}\cdot\mathbf{T}$. Then, the above identity implies that $h_{C\setminus\{c_j\}}(b_i)=H(a_i)$ for every $1\leq i\leq k$. Thus, (\ref{eq3-1-4}) can be further simplified as
\begin{align}
    \mathbf{v}_{h_C\cdot f}\cdot \mathbf{M}\cdot \mathbf{c}_j
    &=-h_{B}(c_j)\left(\sum_{i=1}^{k}\frac{(H\cdot f)(a_i)}{h_{A\setminus\{a_i\}}(a_i)}\right). \label{eq3-1-5}
\end{align}

Meanwhile, by $\deg(f)\leq k-l-1$, we have $\deg(H\cdot f)\leq k-2$. Denote $\tilde{\mathbf{v}}$ as the coefficient vector of $H(x)f(x)$. Then, we have
$$\tilde{\mathbf{v}}\cdot \mathbf{A}^{(k-1)}=((H\cdot f)(a_1),\ldots,(H\cdot f)(a_k)).$$
On the other hand, by Cramer's rule, we know that
$$\mathbf{A}^{(k-1)}\cdot (h_{A\setminus\{a_1\}}(a_1)^{-1},\ldots,h_{A\setminus\{a_k\}}(a_k)^{-1})^{T}=\mathbf{0}.$$
Thus, by (\ref{eq3-1-5}) and
$$\sum_{i=1}^{k}\frac{(H\cdot f)(a_i)}{h_{A\setminus\{a_i\}}(a_i)}=\tilde{\mathbf{v}}\cdot \mathbf{A}^{(k-1)}\cdot (h_{A\setminus\{a_1\}}(a_1)^{-1},\ldots,h_{A\setminus\{a_k\}}(a_k)^{-1})^{T},$$
we have $\mathbf{v}_{h_C\cdot f}\cdot \mathbf{M}\cdot \mathbf{c}_j=0$ for every $f(x)\in\mathbb{F}_q^{<k-l}[x]$ and $1\leq j\leq l$. Note that $\{\mathbf{v}_{h_C\cdot f}: f(x)\in\mathbb{F}_q^{<k-l}[x]\}$ is a $(k-l)$-dim subspace of $\mathbb{F}_q^{k}$ and $\mathbf{v}_{h_C\cdot f}\cdot \mathbf{c}_j=0$ for every $f(x)\in\mathbb{F}_q^{<k-l}[x]$ and $1\leq j\leq l$. Therefore, by $dim(\text{Span}\{\mathbf{c}_1,\ldots,\mathbf{c}_l\})=l$, we have $\mathbf{M}\cdot \mathbf{c}_j\in \text{Span}\{\mathbf{c}_1,\ldots,\mathbf{c}_l\}$. 

This completes the proof.
\end{proof}

Next, using the multiplicative subgroups of $\mathbb{F}_q^*$, we give an explicit construction of $A_i$'s and $C$ that satisfy the requirements in Corollary \ref{coro2-1}.

Let $q\geq (\zeta+1)\cdot\max\{k,l^I\}+1$ be a prime power. Assume that $\max\{k,l^I\}|(q-1)$ and let $G$ be a multiplicative subgroup of $\mathbb{F}_q^{*}$ of order $\max\{k,l^I\}$\footnote{When $\max\{k,l^I\}\nmid(q-1)$, one can let $G$ be the smallest multiplicative subgroup of $\mathbb{F}_q^{*}$ of order at least $\max\{k,l^I\}$ and the condition that $q\geq (\zeta+1)\cdot\max\{k,l^I\}+1$ becomes $q\geq (\zeta+1)\cdot|G|+1$.}. Then, $|\mathbb{F}_q^{*}/G|\geq\zeta+1$. Denote $\alpha\in \mathbb{F}_q^*$ as the generator of the quotient group $\mathbb{F}_q^{*}/G$. Let $A_1=\{a_1,\ldots, a_{k}\}\subseteq G$. Then, we define
\begin{itemize}
  \item $A_i=\alpha^{i-1} A_1=\{\alpha^{i-1}a_1,\ldots, \alpha^{i-1}a_{k}\}$, $2\leq i\leq \zeta+1$;
  \item $C=\{\alpha^{\zeta}a_1,\ldots, \alpha^{\zeta}a_{k}\}\subseteq A_{\zeta+1}$ and $B=\{b_1,\ldots, b_{l^I-l^F}\}$ be any subset of $\alpha^{\zeta}G\setminus C$ of size $l^I-l^F$.
\end{itemize}
Note that $B,C\subseteq \alpha^{\zeta}G$ and $A_i\subseteq \alpha^{i-1} G$, $1\leq i\leq \zeta$. Thus, they are pairwise disjoint in $\mathbb{F}_q^*$. Meanwhile, for $2\leq i\leq \zeta$, 
\begin{align}\label{A_i}
    \mathbf{A}_{i}^{(l^F)}&=\left(\begin{array}{cccc}
       1  & 1 & \cdots & 1 \\
       \alpha^{i-1}a_1 & \alpha^{i-1}a_2 & \cdots & \alpha^{i-1}a_{k} \\
       \vdots & \vdots &  & \vdots \\
       (\alpha^{i-1}a_1)^{l^F-1} & (\alpha^{i-1}a_2)^{l^F-1} & \cdots & (\alpha^{i-1}a_{k})^{l^F-1}
    \end{array}\right) \nonumber \\
    &=\text{diag}(1,\alpha^{i-1},\ldots, \alpha^{(i-1)(l^F-1)})\mathbf{A}_1^{(l^F)}.
\end{align}
Since $\alpha\neq 0$, the diagonal matrix $\text{diag}(1,\alpha^{i-1},\ldots, \alpha^{(i-1)(l^F-1)})$ is invertible. Thus, for $2\leq i\leq \zeta$, set $\mathbf{T}_i=\text{diag}(1,\alpha^{i-1},\ldots, \alpha^{(i-1)(l^F-1)})$ and we have $\mathbf{A}_i^{(l^F)}=\mathbf{T}_i\mathbf{A}_1^{(l^F)}$.

\begin{remark}
The sets $A$, $B$ and $C$ in Example \ref{ex1} can be viewed as an example of the above construction for $q=19$, $k=4$, $l^I=l^F=2$ and $\zeta=2$.
\end{remark}

Clearly, the sets $A_i$'s and $C$ constructed above satisfy the requirements in Corollary \ref{coro2-1}. Then, $A_i$'s, $B$ and $C$ satisfy the requirements in Construction II. This leads to the following immediate corollary.

\begin{corollary}\label{coro2-2}
For positive integers $\zeta\geq 2$, $k$, $l^I$ and $l^F\leq \min\{k,l^I\}$, let $q$ be a prime power such that $\max\{k,l^I\}\mid (q-1)$ and $q\geq (\zeta+1)\cdot\max\{k,l^I\}+1$. Then, there is an explicit construction of $(k+l^I,k;\zeta k+l^F,\zeta k)$ MDS convertible code with optimal access cost.
\end{corollary}

\begin{remark}
The requirement of the field size in Corollary \ref{coro2-2} is $q\geq (\zeta+1)\cdot\max\{k,n^I-k\}+1$, which is less strict than those in \cite{MR22a}.
\end{remark}

\section{Constructions of Locally repairable convertible codes}\label{sec_LRCC}

In this section, we focus on constructions of locally repairable convertible codes (LRCCs) in merge regime. First, we recall the construction of LRCs introduced in \cite{TB14} using ``good'' polynomials. Then, we give a example to illustrate the underlying idea of the construction. In Section \ref{sec_LRCC2}, based on the construction of LRCs using ``good'' polynomials, we present our construction of LRCCs. Though this construction, for $n^I=(k+l^I)(r+1)$ and $n^F=(k+l^F)(r+1)$, we obtain a family of $(n^I,kr,r;n^F,\zeta k r,r)$ LRCCs with write access cost $l^F(r+1)$ and read access cost $\zeta l^Fr$ over finite field of size linear in $n^I$. The access costs of the constructed LRCCs are then shown to be optimal in Section \ref{sec_lb} by Corollary \ref{coro1}. 


\subsection{Optimal LRCs through ``good" polynomials}

To start with, we first recall the general construction of optimal LRCs in \cite{TB14}.

Let $n,k,r$ be positive integers such that $(r+1)|n$ and let $\mathcal{A}=\{A_1,\ldots, A_{\frac{n}{r+1}}\}$ be a family of $n/(r+1)$ mutually disjoint $r+1$-subsets of $\mathbb{F}_q$. A polynomial $g(x)\in \mathbb{F}_q[x]$ of degree $r+1$ is called a \emph{good} polynomial (w.r.t. $\mathcal{A}$) if for each $1\leq i\leq n/(r+1)$, $g$ is constant on $A_i$ and we use $g(A_i)$ to denote this constant. 

Let $\mathbf{m}=(m_1,\ldots,m_k)\in \mathbb{F}_q^{kr}$ be a message vector. Given a good polynomial $g$, define the encoding polynomial
\begin{equation}\label{cons_LRC}
  f_\mathbf{m}(x)=\sum_{i=1}^{r}\sum_{j=0}^{k-1}m_{jr+i}g(x)^jx^{i-1}.
\end{equation}
Then, as shown in \cite{TB14}, $\mathcal{C}_0=\{(f_{\mathbf{m}}(a),a\in A):\mathbf{m}\in\mathbb{F}_q^{kr}\}$
is an optimal $(n,kr,r)$-LRC w.r.t. the Singleton-type bound in Theorem \ref{Singleton_bound}. Let $\mathbf{v}=(v_1,\ldots,v_n)\in (\mathbb{F}_q^{*})^{n}$. Like GRS codes to RS codes, one can easily obtain that the following code
$$\mathcal{C}=\{(v_if_{\mathbf{m}}(a_i),i\in [n]):\mathbf{m}\in\mathbb{F}_q^{kr}\}$$
is also an $(n,kr,r)$-optimal LRC. For simplicity, we call $\mathcal{C}$ the LRC defined by $\mathcal{A}$, $g(x)$ and $\mathbf{v}$.
Specially, when $\mathbf{v}$ is the all one vector, we call $\mathcal{C}$ the LRC defined by $\mathcal{A}$ and $g(x)$. 


Let $\tilde{\mathcal{C}}$ be the $[n,kr]$ GRS code defined by evaluation points $\bigcup_{i=1}^{n/(r+1)}A_i$ and vector $\mathbf{v}$. There is a one-to-one correspondence between codewords of $\tilde{\mathcal{C}}$ and polynomials of degree less than $kr$. Similarly, there is also a one-to-one correspondence between codewords of $\mathcal{C}$ and polynomials in $\mathbb{F}_q[x]$ of form (\ref{cons_LRC}). From this perspective, $\mathcal{C}$ can be viewed as the image of $\tilde{\mathcal{C}}$ under the base transformation that maps $\{x^{i}\}_{0\leq i\leq kr-1}$ to $\{x^{i}g^{j}\}_{0\leq i\leq r-1\atop 0\leq j\leq k-1}$. Based on this connection between the LRCs obtained from good polynomials and the GRS codes, in the following, we generalize our constructions of MDS convertible codes in Section \ref{sec_MDSC} and obtain constructions of optimal LRCCs.

\subsection{An illustrative example}

Similar to Section \ref{sec_MDSC1}, we begin with the following example.

\begin{example}\label{ex2}
In this example, we construct a $(9,4,2)$-LRC $\mathcal{C}^I$ and a $(15,8,2)$-LRC $\mathcal{C}^F$ over $\mathbb{F}_q=\mathbb{F}_{19}$, and a conversion procedure $T$ from $\mathcal{C}^I$ to $\mathcal{C}^F$ with write access cost $3$ and read access cost $4$. Together, they form an $(9,4,2;15,8,2)$ optimal LRCC with optimal access cost.

Let $A_1=\{a_{11},a_{12},a_{13}\}=\{1,7,11\}$, $A_2=\{a_{21},a_{22},a_{23}\}=\{8,18,12\}$, $B_1=\{b_{11},b_{12},b_{13}\}=\{2,14,3\}$, $B_2=\{b_{21},b_{22},b_{23}\}=\{16,17,5\}$ and $C=\{c_1,c_2,c_3\}=\{4,9,6\}$. Let $A=A_1\cup A_2$ and $B=B_1\cup B_2$.
By Table I, it holds that $a_{1j}=2^{6(j-1)}$ ($1\leq j\leq 3$), $A_2=2^3 A_1$, $B_1=2 A_1$, $B_2=2^4 A_1$ and $C=2^2 A_1$. 

Let $g(x)=x^3$. Then, $g$ is constant on each $A_i$, $B_i$ and $C$. Specifically, we have $g(A_1)=1$, $g(A_2)=18$, $g(B_1)=8$, $g(B_2)=11$ and $g(C)=7$. Denote $G_1=\{g(A_1),g(A_2)\}=\{1,18\}$ and $G_2=\{g(B_1),g(B_2)\}=\{8,11\}$.

For $f\in \text{Span}_{\mathbb{F}_q}\{1,x,g(x),xg(x)\}$, denote $\mathbf{v}_f$ as its coefficient vector w.r.t. basis $\{1,x,g(x),xg(x)\}$. Define the initial code $\mathcal{C}^I$ as the $(9,4,2)$-LRC defined by $\{A_1,A_2,C\}$ and $g$, i.e., 
$$\mathcal{C}^I=\{(f(\alpha), \alpha\in A\cup C):~f\in \text{Span}_{\mathbb{F}_q}\{1,x,g(x),xg(x)\}\}.$$

Next, to define the final code $\mathcal{C}^F$, we need some preliminaries. Let $\theta_1=5$, $\theta_2=15$ and $\bm{\theta}=(5,5,15,15)$. Define
\begin{align*} 
\mathbf{M}
&=\left(\begin{array}{cccc}
1 & 1 & 1 & 1\\
a_{11} & a_{12} & a_{21} & a_{22}\\
a_{11}^{3} & a_{12}^{3} & a_{21}^{3} & a_{22}^{3}\\
a_{11}^{4} & a_{12}^{4} & a_{21}^{4} & a_{22}^{4}
\end{array}\right)\cdot 
\text{diag}(\bm{\theta})\cdot
\left(\begin{array}{cccc}
1 & 1 & 1 & 1\\
b_{11} & b_{12} & b_{21} & b_{22}\\
b_{11}^{3} & b_{12}^{3} & b_{21}^{3} & b_{22}^{3}\\
b_{11}^{4} & b_{12}^{4} & b_{21}^{4} & b_{22}^{4}
\end{array}\right)^{-1}\\
&=\left(\begin{array}{cccc}
10 & 0 & 16 & 0\\
0 & 5 & 0 & 8\\
14 & 0 & 6 & 0\\
0 & 7 & 0 & 3
\end{array}\right).
\end{align*}
Then, one can verify that
\begin{itemize}
  \item [1.] $\mathbf{M}\cdot(1,b_{ij},b_{ij}^3,b_{ij}^4)^{T}=\theta_i(1,a_{ij},a_{ij}^3,a_{ij}^4)^{T}$, $(i,j)\in [2]\times [3]$;
  \item [2.] $\mathbf{M}\cdot(1,c_1,c_1^3,c_1^4)^{T}=(8,16,18,17)^{T}=15 (1,c_1,c_1^3,c_1^4)^{T}-7(1,c_2,c_2^3,c_2^4)^{T}$;
  \item [2.] $\mathbf{M}\cdot(1,c_2,c_2^3,c_2^4)^{T}=(8,17,18,5)^{T}=11 (1,c_1,c_1^3,c_1^4)^{T}-3(1,c_2,c_2^3,c_2^4)^{T}$.
\end{itemize}
For $f\in \text{Span}_{\mathbb{F}_q}\{1,x,g(x),xg(x)\}$, denote $M(f)$ as the polynomial in $\text{Span}_{\mathbb{F}_q}\{1,x,g(x),xg(x)\}$ with coefficient vector $\mathbf{v}_f\cdot \mathbf{M}$ w.r.t. basis $\{1,x,g(x),xg(x)\}$. Then,
\begin{equation*}
    \begin{cases}
    M(f)(b_{i,j})=\mathbf{v}_f\cdot \mathbf{M}\cdot (1,b_{ij},b_{ij}^3,b_{ij}^4)^{T}=\theta_if(a_{ij}),~(i,j)\in [2]\times [3];\\
    M(f)(c_1)=\mathbf{v}_f\cdot \mathbf{M}\cdot (1,c_1,c_1^3,c_1^4)^{T}=15f(c_1)-7f(c_2);\\
    M(f)(c_2)=\mathbf{v}_f\cdot \mathbf{M}\cdot (1,c_2,c_2^3,c_2^4)^{T}=11f(c_1)-3f(c_2).
    \end{cases}
\end{equation*}

For $f_1,f_2\in \text{Span}_{\mathbb{F}_q}\{1,x,g(x),xg(x)\}$, define $T$ as the following map over $\mathbb{F}_q[x]$
$$T:~(f_1,f_2)\longmapsto (h_{G_2}\circ g)\cdot f_1+(h_{G_1}\circ g)\cdot M(f_2).$$
Note that $h_{G_1}\circ g=(g(x)-g(A_1))(g(x)-g(A_2))$ and $h_{G_2}\circ g=(g(x)-g(B_1))(g(x)-g(B_2))$. Thus, we have
\begin{equation}\label{ex2-eq1}
    T(f_1,f_2)\in\text{Span}_{\mathbb{F}_q}\{x^ig^{j}\}_{0\leq i\leq 1\atop 0\leq j\leq 3}.
\end{equation} 
Let $\mathbf{u}=(u_\alpha, \alpha\in A\cup B\cup C)$ be a vector in $\mathbb{F}_{q}^{15}$ with coordinates indexed by elements in $A\cup B\cup C$ and 
\begin{equation*}
    \begin{cases}
    u_{a_{ij}}=(h_{G_2}\circ g)(a_{ij})^{-1},~(i,j)\in [2]\times[3];\\
    u_{b_{ij}}=\theta_i^{-1}(h_{G_1}\circ g)(b_{ij})^{-1},~(i,j)\in [2]\times[3];\\
    u_{c_i}=1,~1\leq i\leq 3.\\
    \end{cases}
\end{equation*}

Now, we define the final code $\mathcal{C}^{F}$ as
$$\mathcal{C}^F=\{(u_\alpha T(f_1,f_2)(\alpha), \alpha\in A\cup B\cup C):~f_1,f_2\in \text{Span}_{\mathbb{F}_q}\{1,x,g(x),xg(x)\}\}.$$
By (\ref{ex2-eq1}), $\mathcal{C}^F$ is a subcode of the optimal $(15,8,2)$-LRC defined by $\{A_1,A_2,B_1,B_2,C\}$ and $\mathbf{u}$. Note that
\begin{equation*}
    u_\alpha T(f_1,f_2)(\alpha)=\begin{cases}
    f_1(a_{i,j}),~\text{when}~\alpha=a_{i,j};\\
    f_2(a_{i,j}),~\text{when}~\alpha=b_{i,j}.
    \end{cases} 
\end{equation*}
Thus, we have $(u_\alpha T(f_1,f_2)(\alpha),\alpha\in A)=(f_1(\alpha),\alpha\in A)$ and $(u_\alpha T(f_1,f_2)(\alpha),\alpha\in B)=(f_2(\alpha),\alpha\in A)$. This means that we can use $(u_\alpha T(f_1,f_2)(\alpha),\alpha\in A)$ and $(u_\alpha T(f_1,f_2)(\alpha),\alpha\in B)$ to reconstruct $f_1$ and $f_2$, respectively. Thus, $\dim(\mathcal{C}^F)=8$. Moreover, we have
\begin{equation*}
    \begin{cases}
    u_{c_1} T(f_1,f_2)(c_1)=4f_1(c_1)+10\cdot (15f_2(c_1)-7f_2(c_2)),\\
    u_{c_2} T(f_1,f_2)(c_2)=4f_1(c_2)+10\cdot (11f_2(c_1)-3f_2(c_2))
    \end{cases}
\end{equation*}
and by locality, $u_{c_3} T(f_1,f_2)(c_3)$ is a linear combination of $u_{c_1} T(f_1,f_2)(c_1)$ and $u_{c_2} T(f_1,f_2)(c_2)$. Therefore, for any two codewords $\mathbf{c}_1=(f_1(\alpha), \alpha\in A\cup C)$ and $\mathbf{c}_2=(f_2(\alpha), \alpha\in A\cup C)$ in $\mathcal{C}^I$, the map $T$ induces a conversion procedure from $(\mathbf{c}_1,\mathbf{c}_2)$ to $\mathbf{d}=(u_\alpha T(f_1,f_2)(\alpha),\alpha\in A\cup B\cup C)\in \mathcal{C}^F$ such that 
\begin{itemize}
    \item [1)]  the $6$ symbols in $\mathbf{c}_1|_{A}$ remain in $\mathbf{d}$ as $\mathbf{d}|_{A}$, and the $6$ symbols in $\mathbf{c}_2|_{A}$ remain in $\mathbf{d}$ as $\mathbf{d}|_{B}$;
    \item [2)] the $3$ symbols in $\mathbf{d}|_{C}$ are linear combinations of symbols in $\mathbf{c}_1|_{\{c_1,c_2\}}$ and $\mathbf{c}_2|_{\{c_1,c_2\}}$.
\end{itemize}
Therefore, the conversion procedure $T$ has write access cost $3$ and read access cost $4$.
\end{example}

\begin{remark}\label{rmk3-1}
Example \ref{ex2} provides an overview of the underlying idea of our general construction. 
That is, by selecting proper evaluation points
$A_i=\{a_{i,1},\ldots,a_{i,r+1}\}$, $B_i=\{b_{i,1},\ldots,b_{i,r+1}\}$, $1\leq i\leq k$ and $C_{i}=\{c_{i,1},\ldots,c_{i,r+1}\}$, $1\leq i\leq l$, we can obtain a polynomial $g$ and a surjective map $T:\text{Span}_{\mathbb{F}_q}\{x^sg^t\}_{0\leq s\leq r-1\atop 0\leq t\leq k-1}\times \text{Span}_{\mathbb{F}_q}\{x^sg^t\}_{0\leq s\leq r-1\atop 0\leq t\leq k-1} \rightarrow \text{Span}_{\mathbb{F}_q}\{x^sg^t\}_{0\leq s\leq r-1\atop 0\leq t\leq 2k-1}$ such that the following holds:
\begin{itemize}
    \item [1.] $g$ is a good polynomial, i.e., $\deg(g)=r+1$ and $g$ is constant on each $A_i$, $B_i$ and $C_i$.
    \item [2.] For any $f_1,f_2\in\text{Span}_{\mathbb{F}_q}\{x^sg^t\}_{0\leq s\leq r-1\atop 0\leq t\leq k-1}$, $T(f_1,f_2)(a_{i,j})=\alpha_{i,j}f_1(a_{i,j})$ and $T(f_1,f_2)(b_{i,j})=\beta_{i,j}f_2(a_{i,j})$, where for each $(i,j)\in [k]\times [r+1]$, $\alpha_{i,j}$ and $\beta_{i,j}$ are non-zero constants determined by $A=\bigcup_{i\in [k]}A_i$, $B=\bigcup_{i\in [k]}B_i$ and $C=\bigcup_{i\in [l]}C_i$;
    \item [3.] For any $f_1,f_2\in\text{Span}_{\mathbb{F}_q}\{x^sg^t\}_{0\leq s\leq r-1\atop 0\leq t\leq k-1}$, $T(f_1,f_2)(c_{i,j})=\sum_{(i',j')\in [l]\times [r]}(\gamma_{i,j,i',j'}^{(1)}f_1(c_{i',j'})+\gamma_{i,j,i',j'}^{(2)}f_2(c_{i',j'}))$, where for each $(i,i',j,j')\in [l]\times [l]\times [r]\times [r]$, $\gamma_{i,j,i',j'}^{(1)}$ and $\gamma_{i,j,i',j'}^{(2)}$ are constants determined by $A$, $B$ and $C$.
\end{itemize}
Then, let $\mathcal{C}^{I}$ be the $((k+l)(r+1),kr,r)$-optimal LRC defined by $\{A_1,\ldots,A_k,C_1,\ldots,C_l\}$ and $g$, and define $\mathcal{C}^{F}$ as 
$$\mathcal{C}^F=\{(u_\alpha T(f_1,f_2)(\alpha),\alpha\in A\cup B\cup C):~f_1,f_2\in \text{Span}_{\mathbb{F}_q}\{x^{s}g^t\}_{0\leq s\leq r-1\atop 0\leq t\leq k-1}\},$$
where $u_{a_{i,j}}=\alpha_{i,j}^{-1}$, $u_{b_{i,j}}=\beta_{i,j}^{-1}$ and $u_{c_{i,j}}=1$. Since $T$ is surjective, $\mathcal{C}^F$ is actually the $((2k+l)(r+1),kr,r)$-optimal LRC defined by $\{A_1,\ldots,A_k,B_1,\ldots,B_k,C_1,\ldots,C_l\}$, $g$ and vector $(u_{\alpha},\alpha\in A\cup B\cup C)$.

The map $T$ induces a conversion procedure from $\mathcal{C}^{I}$ to $\mathcal{C}^{F}$, that is, codewords $\mathbf{c}_1=(f_1(\alpha),\alpha\in A\cup C)$ and $\mathbf{c}_2=(f_2(\alpha),\alpha\in A\cup C)$ are mapped to $\mathbf{d}=(u_\alpha T(f_1,f_2)(\alpha),\alpha\in A\cup B\cup C)$, where $f_1,f_2\in\text{Span}_{\mathbb{F}_q}\{x^sg^t\}_{0\leq s\leq r-1\atop 0\leq t\leq k-1}$. Similar to Example \ref{ex2}, one can easily verify that $\mathbf{d}|_{A}=\mathbf{c}_1|_{A}$, $\mathbf{d}|_{B}=\mathbf{c}_2|_{A}$ and symbols in $\mathbf{d}|_{C}$ are linear combinations of symbols in $\mathbf{c}_1|_{\tilde{C}}$ and $\mathbf{c}_2|_{\tilde{C}}$, where $\tilde{C}=\bigcup_{i\in [l]}C_i\setminus \{c_{i,r+1}\}$. This will imply that $(\mathcal{C}^{I},\mathcal{C}^{F})$ is an optimal LRCC with write access cost $l(r+1)$ and read access cost $2lr$.
\end{remark}

\begin{remark}\label{rmk3-2}
In Example \ref{ex2}, the map $T$ is defined as 
$$T:~(f_1,f_2)\longmapsto (h_{G_2}\circ g)\cdot f_1+(h_{G_1}\circ g)\cdot M(f_2),$$
where $G_1=\{g(A_1),\ldots,g(A_k)\}$, $G_2=\{g(B_1),\ldots,g(B_k)\}$ and $M(f_2)\in \text{Span}_{\mathbb{F}_q}\{x^sg^t\}_{0\leq s\leq r-1\atop 0\leq t\leq k-1}$ is the polynomial with coefficient vector $\mathbf{v}_{f_2}\cdot \mathbf{M}$ w.r.t. basis $\{x^sg^t\}_{0\leq s\leq r-1\atop 0\leq t\leq k-1}$. Note that $\mathbf{M}$ is a $kr\times kr$ matrix satisfying: 
\begin{itemize}
    \item [1.] For every $(i,j)\in [k]\times [r+1]$, $\mathbf{M}\cdot \mathbf{b}_{i,j}=\theta_i\mathbf{a}_{i,j}$, where $\mathbf{a}_{i,j}=(a_{i,j}^{s}g^{t}(a_{i,j}),~(s,t)\in [0,r-1]\times [0,k-1])^{T}$, $\mathbf{b}_{i,j}=(b_{i,j}^{s}g^{t}(b_{i,j}),~(s,t)\in [0,r-1]\times [0,k-1])^{T}$ and for each $1\leq i\leq k$, $\theta_i$ is a non-zero constant determined by $A$, $B$ and $C$;
    \item [2.] For every $(i,j)\in [l]\times [r]$, $\mathbf{M}\cdot\mathbf{c}_{i,j}\in \text{Span}_{\mathbb{F}_q}\{\mathbf{c}_{i',j'}:~(i',j')\in [l]\times [r]\}$, where $\mathbf{c}_{i,j}=(c_{i,j}^{s}g^{t}(c_{i,j}),~(s,t)\in [0,r-1]\times [0,l-1])^{T}$.
\end{itemize}
As we shall prove Section \ref{sec_LRCC2}, these properties of $\mathbf{M}$ ensure the map $T$ defined above satisfying the requirements in Remark \ref{rmk3-1}. Thus, the crucial part of the construction is to find proper subsets $A$, $B$ and $C$ such that there exists a $kr\times kr$ matrix $\mathbf{M}$ satisfying the above requirements.
\end{remark}

\subsection{Constructions of LRCCs with optimal access costs}\label{sec_LRCC2}

Throughout the section, we assume that $n^I=(k+l^I)(r+1)$, $n^F=(k+l^F)(r+1)$ and $l^F\leq \min\{k,l^I\}$. Similar to Section \ref{sec_MDSC2}, we first assume that $l^F=l^I=l\leq k$ and present a construction of optimal LRCCs for this case. Then, we modify it for the general case when $l^F\leq l^I\leq k$.

\textbf{Construction III}: Let  $A_{s,i}=\{a_{s,i,1},\ldots,a_{s,i,r+1}\}$, $(s,i)\in [\zeta]\times [k]$ and $C_{i}=\{c_{i,1},\ldots,c_{i,r+1}\}$, $1\leq i\leq l$, be mutually disjoint $(r+1)$-subsets of $\mathbb{F}_q$ such that the following holds:
\begin{description}
    \item [Condition 1] There is a polynomial $g(x)\in \mathbb{F}_q[x]$ of degree $r+1$ such that $g$ is constant on each $(s,i)\in [\zeta]\times [k]$ and $C_{i}$, $1\leq i\leq l$. Denote $G_0=\{g({C_1}),\ldots,g({C_l})\}$ and $G_s=\{g({A_{s,1}}),\ldots,g({A_{s,k}})\}$.
    \item [Condition 2] For every $2\leq s\leq \zeta$, there is a $kr\times kr$ matrix $\mathbf{M}_s$ satisfying:
    \begin{equation}\label{cond_LRCC}
    \begin{cases}
        \mathbf{M}_s\cdot \mathbf{a}_{s,i,j}=\theta_{s,i} \mathbf{a}_{1,i,j},~(i,j)\in [k]\times [r+1];\\
        \mathbf{M}_s\cdot\mathbf{c}_{i,j}\in \text{Span}_{\mathbb{F}_q}\{\mathbf{c}_{i',j'}:~(i',j')\in [l]\times [r]\},~(i,j)\in [l]\times [r],
    \end{cases}
    \end{equation}
    where $\mathbf{a}_{s,i,j}=(a_{s,i,j}^{t_1}g^{t_2}(a_{s,i,j}),~(t_1,t_2)\in [0,r-1]\times [0,k-1])^{T}$, $\mathbf{c}_{i,j}=(c_{i,j}^{t_1}g^{t_2}(c_{i,j}),~(t_1,t_2)\in [0,r-1]\times [0,l-1])^{T}$ and $$\theta_{s,i}=\frac{h_{G_s\cup G_0 \setminus\{g(A_{s,i})\}}(g(A_{s,i}))}{h_{G_1\cup G_0\setminus\{g(A_{1,i})\}}(g(A_{1,i}))}.$$
\end{description}
Denote $A_s=\bigcup_{i=1}^{k}A_{s,i}$, $C=\bigcup_{i=1}^{l}C_i$, $G=\bigcup_{s=1}^{\zeta}G_s$ and $V_{k,r}=\text{Span}_{\mathbb{F}_q}\{x^{t_1}g^{t_2}\}_{0\leq t_1\leq r-1\atop 0\leq t_2\leq k-1}$. Then, define $\mathcal{C}^I$ as
\begin{equation*}
    \mathcal{C}^I=\{(f(\alpha),~\alpha\in A_1\cup C):~f\in V_{k,r}\}.
\end{equation*}
Let $\mathbf{M}_1=\mathbf{I}$ and define the map $T$ as
\begin{align}\label{cons3-1}
   T: ~~~~~~\left(V_{k,r}\right)^{\zeta}&\longrightarrow V_{\zeta k,r} \nonumber\\
   (f_{1},\ldots,f_{\zeta})&\longmapsto \sum_{s=1}^{\zeta}(h_{G\setminus G_s}\circ g)\cdot M_s(f_s),
\end{align}
where $M_s(f_s)\in V_{k,r}$ is the polynomial with coefficient vector $\mathbf{v}_{f_s}\cdot\mathbf{M}_{s}$ w.r.t. basis $\{x^{t_1}g^{t_2}\}_{0\leq t_1\leq r-1\atop 0\leq t_2\leq k-1}$. Then, we define the final code $\mathcal{C}^F$ as
\begin{align}\label{cons3-2}
  \mathcal{C}^F =\{(u_{\alpha}T(f_{1},\ldots,f_{\zeta})(\alpha),~\alpha\in (\bigcup_{s=1}^{\zeta}A_s)\cup C): f_s\in V_{k,r},~1\leq s\leq \zeta\},
\end{align}
where $u_{\alpha}=\theta_{s,i}^{-1}h_{G\setminus G_s}^{-1}(g(A_{s,i}))$, when $\alpha=a_{s,i,j}$ and $u_{\alpha}=1$, when $\alpha=c_{i,j}$.

\begin{theorem}\label{thm4}
For positive integers $\zeta\geq 2$, $r$, $k$ and $l^F=l^I=l\leq k$, the $(\mathcal{C}^I, \mathcal{C}^{F})$ given by Construction III is an $(n^I,kr,r;n^F,\zeta kr,r)$ optimal LRCC with write access cost $l(r+1)$ and read access cost $\zeta lr$.
\end{theorem}

\begin{proof}
Clearly, $\mathcal{C}^{I}$ is the $(n^I,kr,r)$-optimal LRC defined by $\{A_{1,i}\}_{i=1}^{k}\cup \{C_i\}_{i=1}^{l}$ and $g$. Let $\mathbf{c}_1,\ldots,\mathbf{c}_{\zeta}\in \mathcal{C}^{I}$ be the $\zeta$ initial codewords and denote $f_i(x)\in V_{k,r}$ as the encoding polynomial of $\mathbf{c}_i$. Since for every $1\leq s\leq \zeta$, $M_s(f_s)\in V_{k,r}$ and 
$$(h_{G\setminus G_s}\circ g)(x)=\prod_{\alpha\in G\setminus G_s}(g(x)-\alpha)\in \text{Span}_{\mathbb{F}_q}\{g^{i}\}_{0\leq i\leq (\zeta-1) k},$$ 
we have $T(f_1,\ldots,f_{\zeta})\in V_{\zeta k,r}$. Thus, $\mathcal{C}^{F}$ is a subcode of the following $(n^F,\zeta kr,r)$-optimal LRC:
$$\mathcal{C}_0=\{(u_{\alpha}f(\alpha),~\alpha\in (\bigcup_{s=1}^{\zeta}A_s)\cup C):~f\in V_{\zeta k,r}\}.$$
Moreover, the map $T$ induces a conversion procedure from $\mathcal{C}^{I}$ to $\mathcal{C}^{F}$, that is, $\mathbf{c}_1,\ldots,\mathbf{c}_{\zeta}$ are converted to the following codeword
$$\mathbf{d}=(u_{\alpha}T(f_1,\ldots,f_{\zeta})(\alpha),~\alpha\in (\bigcup_{s=1}^{\zeta}A_s)\cup C).$$ 

Next, we show that the read and write access costs of the conversion procedure induced by $T$ are $\zeta lr$ and $l(r+1)$, respectively. As a byproduct, we shall see that $\mathcal{C}^{F}=\mathcal{C}_0$, which yields the result.

For each $(s,i,j)\in [\zeta]\times [k]\times [r+1]$ and $\alpha=a_{s,i,j}$, by $g(a_{s,i,j})=g(A_{s,i})$ and $h_{G\setminus G_s}(x)=\prod_{w\in G\setminus G_s}(x-w)$, we have 
\begin{align*}
    u_{\alpha}T(f_{1},\ldots,f_{\zeta})(\alpha)
    &=\theta_{s,i}^{-1}h_{G\setminus G_s}^{-1}(g(A_{s,i}))\sum_{s'=1}^{\zeta}h_{G\setminus G_{s'}}(g(A_{s,i}))(M_{s'}(f_{s'}))(a_{s,i,j})\\
    &=\theta_{s,i}^{-1}(M_{s}(f_{s}))(a_{s,i,j})\\
    &=f_{s}(a_{1,i,j}),
\end{align*}
where the last equality follows from $(M_{s}(f_{s}))(a_{s,i,j})=\mathbf{v}_{f_s}\cdot \mathbf{M}_s\cdot \mathbf{a}_{s,i,j}=\theta_{s,i}f_{s}(a_{1,i,j})$. This leads to $$\mathbf{d}|_{A_s}=\mathbf{c}_s|_{A_1}=(f_{s}(\alpha),~\alpha\in A_1).$$ 
Thus, for each $1\leq s\leq \zeta$, there are at least $k(r+1)$ symbols in $\mathbf{c}_i$ that remains in $\mathbf{d}$. This implies that the write access cost is at most $l(r+1)$. Moreover, by $\mathcal{C}^{I}|_{A_1}=\mathbb{F}_q^{kr}$, $\mathbf{d}|_{A_s}=\mathbf{c}_s|_{A_1}$ also implies that $\mathcal{C}^{F}|_{\bigcup_{s=1}^{\zeta}A_s}=\mathbb{F}_q^{\zeta kr}$. In other words,  $\dim(\mathcal{C}^{F})\geq \zeta kr$. By $\mathcal{C}^{F}\subseteq \mathcal{C}_0$, this leads to $\mathcal{C}^{F}=\mathcal{C}_0$.

On the other hand, for each $(i,j)\in [l]\times [r]$ and $\alpha=c_{i,j}$, we have
\begin{align}
    u_{\alpha}T(f_{1},\ldots,f_{\zeta})(\alpha)&=\sum_{s=1}^{\zeta}h_{G\setminus G_{s}}(g(C_i))(M_{s}(f_{s}))(c_{i,j}).\label{eq3-5-1}
\end{align}
Since $\mathbf{M}_s\cdot\mathbf{c}_{i,j}\in \text{Span}_{\mathbb{F}_q}\{\mathbf{c}_{i',j'}:~(i',j')\in [l]\times [r]\}$, we can assume that $\mathbf{M}_{s}\cdot \mathbf{c}_{i,j}=\sum_{(i',j')\in [l]\times [r]}\eta_{s,i',j'}\mathbf{c}_{i',j'}$. Therefore,
\begin{align}
    (M_s(f_s))(c_{i,j})&=\mathbf{v}_{f_s}\cdot\left(\sum_{(i',j')\in [l]\times [r]}\eta_{s,i',j'}\mathbf{c}_{i',j'}\right)\nonumber\\
   &=\sum_{(i',j')\in [l]\times [r]}\eta_{s,i',j'}f_s(c_{i',j'}).\label{eq3-5-2}
\end{align}
Denote $\tilde{C}=\bigcup_{i\in [l]}C_i\setminus \{c_{i,r+1}\}$. Note that $f_{s}(c_{i,j})=\mathbf{c}_{s}|_{\{c_{i,j}\}}$. Thus, (\ref{eq3-5-1}) and (\ref{eq3-5-2})
implies that for every $1\leq j\leq l$, symbols in $\mathbf{d}|_{C_{j}\setminus\{c_{j,r+1}\}}$ are linear combinations of the $\zeta lr$ symbols in $\mathbf{c}_1|_{\tilde{C}},\ldots,\mathbf{c}_{\zeta}|_{\tilde{C}}$. Moreover, by locality, $\mathbf{d}|_{\{c_{j,r+1}\}}$ can be recovered by the $r$ symbols in $\mathbf{d}|_{C_{j}\setminus\{c_{j,r+1}\}}$. Therefore, all the $l(r+1)$ symbols of $\mathbf{d}|_{C}$ can be constructed by the $\zeta lr$ symbols in $\mathbf{c}_1|_{\tilde{C}},\ldots,\mathbf{c}_{\zeta}|_{\tilde{C}}$. Thus, the read access cost is at most $\zeta lr$.
\end{proof}

Now, we present the construction for the case when $l^F\leq l^I$, which follows the same idea as Construction II.

\textbf{Construction IV}: Let $A_{s,i}$, $A_{s}$, $C_i$ and $g(x)$ be identical to those defined in Construction III. For $1\leq i\leq l^I-l^F$, let $B_{i}=\{b_{i,1},\ldots,b_{i,r+1}\}$ be mutually disjoint subsets of size $r+1$ in $\mathbb{F}_q\setminus A_1\cup C$ such that $g$ is constant on each $B_i$. Denote $B=\bigcup_{i=1}^{l^I-l^F}B_i$.\footnote{When $l^I=l^F$, $B=\emptyset$} Define $\mathcal{C}^I$ as the following optimal $(n^I,kr,r)$-LRC:
\begin{equation*}
    \mathcal{C}^I=\{(f(\alpha),~\alpha\in A_1\cup C\cup B): f\in V_{k,r}\}.
\end{equation*}
Let $T$ and the final code $\mathcal{C}^F$ be the same as those defined in (\ref{cons3-1}) and (\ref{cons3-2}). 

Under this construction, symbols indexed by evaluation points $B$ of codewords in $\mathcal{C}^{I}$ don't participant the conversion procedure induced by $T$. Thus, by Theorem \ref{thm4}, we have the following immediate result.

\begin{theorem}\label{thm5}
For positive integers $\zeta\geq 2$, $r$, $k$, $l^I$ and $l^F\leq \min\{k,l^I\}$, the $(\mathcal{C}^I, \mathcal{C}^{F})$ given by Construction IV is an $(n^I,kr,r;n^F,\zeta kr,r)$ optimal LRCC with write access cost $l^{F}(r+1)$ and read access cost $\zeta l^{F}r$.
\end{theorem}

\subsection{A family of LRCCs with optimal access costs}

From Section \ref{sec_LRCC2}, the key to implementing Construction III and IV lies in finding proper subsets $A_{s}$'s and $C$ from $\mathbb{F}_q$, such that there exist matrices $\mathbf{M}_s$'s satisfying conditions 1 and 2 in Construction III. Like Section \ref{sec_MDSC3}, in the following, we first give a sufficient condition for when $A_s$'s and $C$ can ensure the existence of such $\mathbf{M}_s$'s. Then, we give an explicit construction of $A_{s,i}$ and $C_i$ satisfying this sufficient condition. Together with proper subset $B$, we provide an implementation of Construction IV, and obtain a family of LRCCs with optimal access costs.

\begin{lemma}\label{lem3-2}
For positive integers $\zeta\geq 2$, $r$, $k$ and $l\leq k$, let $A_{i}=\{a_{i,1},\ldots,a_{i,r+1}\}$, $B_{i}=\{b_{i,1},\ldots,b_{i,r+1}\}$, $1\leq i\leq k$, and $C_{i}=\{c_{i,1},\ldots,c_{i,r+1}\}$, $1\leq i\leq l$, be mutually disjoint $(r+1)$-subsets of $\mathbb{F}_q$. Let $g$ be a polynomial of degree $r+1$ such that $g$ is constant on each $A_i$, $B_i$ and $C_i$. Denote $G_0=\{g(C_1),\ldots,g(C_l)\}$, $G_1=\{g(A_1),\ldots,g(A_k)\}$ and $G_2=\{g(B_1),\ldots,g(B_k)\}$. Suppose that
\begin{itemize}
    \item [i)] there is an invertible matrix $\mathbf{T}$ of order $r$ such that $\mathbf{B}_i^{(r)}=\mathbf{T}\mathbf{A}_{i}^{(r)}$ for every $1\leq i\leq k$;
    \item [ii)] there is an invertibel matrix $\mathbf{T}'$ of order $l$ such that $\mathbf{G}_2^{(l)}=\mathbf{T}'\mathbf{G}_{1}^{(l)}$.
\end{itemize}
Then, there is a $kr\times kr$ matrix $\mathbf{M}$ satisfying:
    \begin{equation*}
    \begin{cases}
        \mathbf{M}\cdot \mathbf{b}_{i,j}=\theta_{i} \mathbf{a}_{i,j},~(i,j)\in [k]\times [r+1];\\
        \mathbf{M}\cdot\mathbf{c}_{i,j}\in \text{Span}_{\mathbb{F}_q}\{\mathbf{c}_{i',j'}:~(i',j')\in [l]\times [r]\},~(i,j)\in [l]\times [r],
    \end{cases}
    \end{equation*}
    where $\mathbf{a}_{i,j}=(a_{i,j}^{s}g^{t}(a_{i,j}),~(s,t)\in [0,r-1]\times [0,k-1])^{T}$, $\mathbf{b}_{i,j}=(b_{i,j}^{s}g^{t}(b_{i,j}),~(s,t)\in [0,r-1]\times [0,k-1])^{T}$, $\mathbf{c}_{i,j}=(c_{i,j}^{s}g^{t}(c_{i,j}),~(s,t)\in [0,r-1]\times [0,l-1])^{T}$ and $$\theta_i=\frac{h_{G_2\cup G_0 \setminus\{g(B_i)\}}(g(B_i))}{h_{G_1\cup G_0\setminus\{g(A_i)\}}(g(A_i))}.$$
\end{lemma}

By Lemma \ref{lem3-2}, we have the following immediate corollary, which provides a sufficient condition for the existence of $\mathbf{M}_s$'s satisfying condition 2 in Construction III.

\begin{corollary}\label{coro3-1}
Let $\zeta \geq 2$, $r$, $k$ and $l<k$ be positive integers. Let $A_{s,i}=\{a_{s,i,1},\ldots,a_{s,i,r+1}\}$, $(s,i)\in [\zeta]\times [k]$ and $C_{i}=\{c_{i,1},\ldots,c_{i,r+1}\}$, $1\leq i\leq l$, be mutually disjoint $(r+1)$-subsets of $\mathbb{F}_q$. Let $g$ be a polynomial of degree $r+1$ such that $g$ is constant on each $A_{s,i}$ and $C_i$. Denote $G_0=\{g({C_1}),\ldots,g({C_l})\}$ and $G_s=\{g(A_{s,1}),\ldots,g(A_{s,k})\}$, $1\leq s\leq \zeta$. If for every $2\leq s\leq \zeta$, \begin{itemize}
  \item [i)] there is an invertible matrix $\mathbf{T}_s$ of order $r$ such that $\mathbf{A}_{s,i}^{(r)}=\mathbf{T}_s \mathbf{A}_{1,i}^{(r)}$ for every $1\leq i\leq k$;
  \item [ii)] there is an invertible matrix $\mathbf{T}_s'$ of order $l$ such that $\mathbf{G}_s^{(l)}=\mathbf{T}_s' \mathbf{G}_1^{(l)}$.
\end{itemize}
Then, for every $2\leq s\leq \zeta$, there is a $kr\times kr$ matrix $\mathbf{M}_s$ satisfying condition 2 in Construction III.
\end{corollary}

Now, we present the proof of Lemma \ref{lem3-2}.

\begin{proof}[Proof of Lemma \ref{lem3-2}]
Denote $V_{k,r}=\text{Span}_{\mathbb{F}_q}\{x^{s}g^{t}\}_{0\leq s\leq r-1\atop 0\leq t\leq k-1}$.
Denote $g_{1,i}=g(A_i)$, $g_{2,i}=g(B_i)$ and $g_{0,i}=g(C_i)$. Let $\tilde{A}_i=A_i\setminus\{a_{i,r+1}\}$ and $\tilde{B}_i=B_i\setminus\{b_{i,r+1}\}$. By $1)$ in Lemma \ref{lem-a1}, $\{\mathbf{a}_{i,j}\}_{1\leq i\leq k \atop 1\leq j\leq r}$ and $\{\mathbf{b}_{i,j}\}_{1\leq i\leq k \atop 1\leq j\leq r}$ form two bases of $\mathbb{F}_q^{kr}$ and $\text{Span}_{\mathbb{F}_q}\{\mathbf{c}_{i,j}:~(i,j)\in [l]\times [r]\}$ has dimension $lr$. Moreover, by $2)$ in Lemma \ref{lem-a1}, we also have
\begin{align}
    \mathbf{a}_{i,r+1}&=\sum_{j=1}^{r}\frac{h_{\tilde{A}_i\setminus\{a_{i,j}\}}(a_{i,r+1})}{h_{\tilde{A}_i\setminus\{a_{i,j}\}}(a_{i,j})}\mathbf{a}_{i,j},\label{eq3-2-3a}\\
    \mathbf{b}_{i,r+1}&=\sum_{j=1}^{r}\frac{h_{\tilde{B}_i\setminus\{b_{i,j}\}}(b_{i,r+1})}{h_{\tilde{B}_i\setminus\{b_{i,j}\}}(b_{i,j})}\mathbf{b}_{i,j}\label{eq3-2-3}
\end{align}
for every $1\leq i\leq k$ and 
\begin{equation}\label{eq3-2-4}
     \mathbf{c}_{s,t}=\sum_{(i,j)\in [k]\times [r]}
     \frac{h_{G_2\setminus\{g_{2,i}\}}(g_{0,s})}{h_{G_2\setminus\{g_{2,i}\}}(g_{2,i})}
     \frac{h_{\tilde{B}_i\setminus\{b_{i,j}\}}(c_{s,t})}{h_{\tilde{B}_i\setminus\{b_{i,j}\}}(b_{i,j})}
     \mathbf{b}_{i,j}
\end{equation}
for every $(s,t)\in [l]\times [r]$.

Denote $\mathbf{A}=(\mathbf{a}_{1,1},\ldots,\mathbf{a}_{1,r},\ldots,\mathbf{a}_{k,1},\ldots,\mathbf{a}_{k,r})$ and $\mathbf{B}=(\mathbf{b}_{1,1},\ldots,\mathbf{b}_{1,r},\ldots,\mathbf{b}_{k,1},\ldots,\mathbf{b}_{k,r})$ as the $kr\times kr$ matrices with $(i-1)r+j$-th column $\mathbf{a}_{i,j}$ and $\mathbf{b}_{i,j}$, respectively. Define vector $$\bm{\theta}=(\underbrace{\theta_1,\ldots,\theta_1}_{r},\underbrace{\theta_2,\ldots,\theta_2}_{r},\ldots, \underbrace{\theta_k,\ldots,\theta_k}_{r}).$$ Then, we define matrix $\mathbf{M}$ as
\begin{align}\label{eq3-2-5}
    \mathbf{M}=\mathbf{A}\cdot \text{diag}(\bm{\theta})\cdot \mathbf{B}^{-1}.
\end{align}
Next, we show that $\mathbf{M}$ satisfies the two properties claimed in the statement of Lemma \ref{lem3-2}. 

Clearly, for every $(i,j)\in [k]\times [r]$, $\mathbf{M}\cdot \mathbf{b}_{i,j}=\theta_i\mathbf{a}_{i,j}$. To confirm the first property, we only need to show that $\mathbf{M}\cdot \mathbf{b}_{i,r+1}=\theta_i\mathbf{a}_{i,r+1}$ holds for every $1\leq i\leq k$. By $\mathbf{B}_i^{(r)}=\mathbf{T}\mathbf{A}_{i}^{(r)}$, we have $(1,\ldots,b_{i,r+1}^{r-1})^T=\mathbf{T}\cdot(1,\ldots,a_{i,r+1}^{r-1})^T$ and $\mathbf{T}=\tilde{\mathbf{B}}_i\tilde{\mathbf{A}}_i^{-1}$, where $\tilde{\mathbf{A}}_i$ and $\tilde{\mathbf{B}}_{i}$ are the $r\times r$ Vandermonde matrix generated by $\tilde{A}_{i}$ and $\tilde{B}_i$, respectively. Thus, $\tilde{\mathbf{B}}_i^{-1}\cdot (1,\ldots,b_{i,r+1}^{r-1})^{T}=\tilde{\mathbf{A}}_i^{-1}\cdot (1,\ldots,a_{i,r+1}^{r-1})^{T}$. Note that for every $1\leq i\leq k$ and $x\in \mathbb{F}_q$,
\begin{align}\label{eq3-2-5a}
   \tilde{\mathbf{A}}_i^{-1}\cdot (1,\ldots,x^{r-1})^{T}&=(\frac{h_{\tilde{A}_i\setminus\{a_{i,1}\}}(x)}{h_{\tilde{A}_i\setminus\{a_{i,1}\}}(a_{i,1})},\ldots,\frac{h_{\tilde{A}_i\setminus\{a_{i,r}\}}(x)}{h_{\tilde{A}_i\setminus\{a_{i,r}\}}(a_{i,r})})^{T},\nonumber\\
   \tilde{\mathbf{B}}_i^{-1}\cdot (1,\ldots,x^{r-1})^{T}&=(\frac{h_{\tilde{B}_i\setminus\{b_{i,1}\}}(x)}{h_{\tilde{B}_i\setminus\{b_{i,1}\}}(b_{i,1})},\ldots,\frac{h_{\tilde{B}_i\setminus\{b_{i,r}\}}(x)}{h_{\tilde{B}_i\setminus\{b_{i,r}\}}(b_{i,r})})^{T}.
\end{align}
Thus, $\frac{h_{\tilde{A}_i\setminus\{a_{i,j}\}}(a_{i,r+1})}{h_{\tilde{A}_i\setminus\{a_{i,j}\}}(a_{i,j})}=\frac{h_{\tilde{B}_i\setminus\{b_{i,j}\}}(b_{i,r+1})}{h_{\tilde{B}_i\setminus\{b_{i,j}\}}(b_{i,j})}$ for every $1\leq j\leq r$. By (\ref{eq3-2-3}) and (\ref{eq3-2-5}), this implies that
\begin{align*}   
\mathbf{M}\cdot \mathbf{b}_{i,r+1}&=\theta_i\sum_{j=1}^{r}\frac{h_{\tilde{B}_i\setminus\{b_{i,j}\}}(b_{i,r+1})}{h_{\tilde{B}_i\setminus\{b_{i,j}\}}(b_{i,j})}\mathbf{a}_{i,j}\\
&=\theta_i\sum_{j=1}^{r}\frac{h_{\tilde{A}_i\setminus\{a_{i,j}\}}(a_{i,r+1})}{h_{\tilde{A}_i\setminus\{a_{i,j}\}}(a_{i,j})}\mathbf{a}_{i,j}.
\end{align*}
Thus, $\mathbf{M}\cdot \mathbf{b}_{i,r+1}=\theta_i\mathbf{a}_{i,r+1}$ follows from (\ref{eq3-2-3a}). This confirms the first property.

Next, we show that for each $(i,j)\in [l]\times [r]$, $\mathbf{M}\cdot\mathbf{c}_{i,j}\in \text{Span}_{\mathbb{F}_q}\{\mathbf{c}_{i',j'}:~(i',j')\in [l]\times [r]\}$.

Define $H_C(x)=(h_{G_0}\circ g)(x)=\prod_{i=1}^{l}\left(g(x)-g_{0,i}\right)$. Then, for every $f\in V_{k-l,r}$, $H_C\cdot f\in V_{k,r}$.
Denote $\mathbf{v}_{H_C\cdot f}\in\mathbb{F}_q^{kr}$ as the coefficient vector of $H_C\cdot f$ w.r.t. basis $\{x^{s}g^t\}_{0\leq s\leq r-1\atop 0\leq t\leq k-1}$. Then, for every $(s,t)\in [l]\times [r]$, by (\ref{eq3-2-4}), we have 
\begin{align}
    \mathbf{v}_{H_C\cdot f}\cdot \mathbf{M}\cdot \mathbf{c}_{s,t}&=\sum_{(i,j)\in [k]\times [r]}
    \frac{h_{G_2\setminus\{g_{2,i}\}}(g_{0,s})}{h_{G_2\setminus\{g_{2,i}\}}(g_{2,i})}
    \frac{h_{\tilde{B}_i\setminus\{b_{i,j}\}}(c_{s,t})}{h_{\tilde{B}_i\setminus\{b_{i,j}\}}(b_{i,j})}
    \theta_i
    (\mathbf{v}_{H_C\cdot f}\cdot \mathbf{a}_{i,j}) \nonumber\\
    &=\sum_{(i,j)\in [k]\times [r]}\frac{h_{G_2\setminus\{g_{2,i}\}}(g_{0,s})}{h_{G_2\setminus\{g_{2,i}\}}(g_{2,i})}
    \frac{h_{\tilde{B}_i\setminus\{b_{i,j}\}}(c_{s,t})}{h_{\tilde{B}_i\setminus\{b_{i,j}\}}(b_{i,j})}
    \theta_i
    (H_C\cdot f)({a}_{i,j}) \label{eq3-2-6}. 
\end{align}
Since $\theta_i=\frac{h_{G_2\cup G_0 \setminus\{g_{2,i}\}}(g_{2,i})}{h_{G_1\cup G_0\setminus\{g_{1,i}\}}(g_{1,i})}$ and $H_C(a_{i,j})=h_{G_0}(g_{1,i})$, (\ref{eq3-2-6}) leads to
\begin{align}
   \mathbf{v}_{H_C\cdot f}\cdot \mathbf{M}\cdot \mathbf{c}_{s,t}&=\sum_{(i,j)\in [k]\times [r]}
    \frac{h_{\tilde{B}_i\setminus\{b_{i,j}\}}(c_{s,t})}{h_{\tilde{B}_i\setminus\{b_{i,j}\}}(b_{i,j})}
    \frac{h_{G_0}(g_{2,i})}{h_{G_1\setminus\{g_{1,i}\}}(g_{1,i})}
    h_{G_2\setminus\{g_{2,i}\}}(g_{0,s})
    f({a}_{i,j})\nonumber \\
    &=-h_{G_2}(g_{0,s})
    \left(\sum_{(i,j)\in [k]\times [r]}
    \frac{h_{\tilde{B}_i\setminus\{b_{i,j}\}}(c_{s,t})}{h_{\tilde{B}_i\setminus\{b_{i,j}\}}(b_{i,j})}
    \frac{h_{G_0\setminus\{g_{0,s}\}}(g_{2,i})}{h_{G_1\setminus\{g_{1,i}\}}(g_{1,i})}
    f({a}_{i,j})\right)\label{eq3-2-7}, 
\end{align}
where the 2nd equality follows from $$h_{G_0}(g_{2,i})h_{G_2\setminus\{g_{2,i}\}}(g_{0,s})=-h_{G_2}(g_{0,s})h_{G_0\setminus\{g_{0,s}\}}(g_{2,i}).$$
For $1\leq s\leq l$, denote $\mathbf{h}_{G_0\setminus\{g_{0,s}\}}$ as the coefficient vector of $h_{G_0\setminus\{g_{0,s}\}}(x)$ w.r.t. basis $\{x^{i}\}_{i=0}^{l-1}$ and define $H_s(x)=\mathbf{h}_{G_0\setminus\{g_{0,s}\}}\cdot\mathbf{T}'\cdot (1,x,\ldots,x^{l-1})^{T}$. 
By $\mathbf{G}_2^{(l)}=\mathbf{T}'\mathbf{G}_{1}^{(l)}$, we have 
$$\mathbf{h}_{G_0\setminus\{g_{0,s}\}}\cdot\mathbf{G}_{2}^{(l)}=\mathbf{h}_{G_0\setminus\{g_{0,s}\}}\cdot\mathbf{T'}\mathbf{G}_{1}^{(l)}.$$ 
Thus, $h_{G_0\setminus\{g_{0,s}\}}(g_{2,i})=H_s(g_{1,i})$ for every $1\leq i\leq k$. For $1\leq i\leq k$, let $$\mathbf{u}_i=\mathbf{T}^{-1}\tilde{\mathbf{A}}_i\cdot(\frac{h_{\tilde{A}_i\setminus\{a_{i,1}\}}(c_{s,t})}{h_{\tilde{A}_i\setminus\{a_{i,1}\}}(a_{i,1})},\ldots,\frac{h_{\tilde{A}_i\setminus\{a_{i,r}\}}(c_{s,t})}{h_{\tilde{A}_i\setminus\{a_{i,r}\}}(a_{i,r})})^{T}.$$
By (\ref{eq3-2-5a}), we have $\mathbf{u}_i=\mathbf{T}^{-1}\cdot(1,\ldots,c_{s,t}^{r-1})^{T}$. Note that, $\mathbf{T}=\tilde{\mathbf{B}}_i\tilde{\mathbf{A}}_i^{-1}$, this leads to
\begin{align*}
\tilde{\mathbf{A}}_i^{-1}\cdot \mathbf{u}_i&=\tilde{\mathbf{A}}_i^{-1}\mathbf{T}^{-1}\tilde{\mathbf{A}}_i\cdot(\frac{h_{\tilde{A}_i\setminus\{a_{i,1}\}}(c_{s,t})}{h_{\tilde{A}_i\setminus\{a_{i,1}\}}(a_{i,1})},\ldots,\frac{h_{\tilde{A}_i\setminus\{a_{i,r}\}}(c_{s,t})}{h_{\tilde{A}_i\setminus\{a_{i,r}\}}(a_{i,r})})^{T}\\
&=\tilde{\mathbf{B}}_i^{-1}\cdot(1,\ldots,c_{s,t}^{r-1})^{T}\\
&=(\frac{h_{\tilde{B}_i\setminus\{b_{i,1}\}}(c_{s,t})}{h_{\tilde{B}_i\setminus\{b_{i,1}\}}(b_{i,1})},\ldots,\frac{h_{\tilde{B}_i\setminus\{b_{i,r}\}}(c_{s,t})}{h_{\tilde{B}_i\setminus\{b_{i,r}\}}(b_{i,r})})^{T}.
\end{align*}
Denote $\tilde{H}_s(x)=(H_s\circ g)(x)$. Then, $(\tilde{H}_s\cdot f)(a_{i,j})=H_s(g_{1,i})f(a_{i,j})=h_{G_0\setminus\{g_{0,s}\}}(g_{2,i})f(a_{i,j})$. Thus, the RHS of (\ref{eq3-2-7}) can be simplified as
\begin{align}
    -h_{G_2}(g_{0,s})
    \sum_{i\in [k]}
    ((\tilde{H}_s\cdot f)(a_{i,1}),\ldots,(\tilde{H}_s\cdot f)(a_{i,r}))\cdot
    \frac{\tilde{\mathbf{A}}_i^{-1}\cdot \mathbf{u}_i}{h_{G_1\setminus\{g_{1,i}\}}(g_{1,i})}\label{eq3-2-8}.
\end{align}

Note that $\deg(H_s)\leq l-1$ and $f\in V_{k-l,r}$. Thus, $\tilde{H}_s\cdot f\in \mathbf{V}_{k-1,r}$. Denote $\tilde{\mathbf{v}}\in \mathbb{F}_{q}^{(k-1)r}$ as the coefficient vector of $\tilde{H}_s\cdot f$ w.r.t. basis $\{x^sg^{t}\}_{0\leq s\leq k-2 \atop 0\leq t\leq r-1}$.
Then, for $1\leq i\leq k$, we have
\begin{equation*}
    \tilde{\mathbf{v}}\cdot \left(\begin{array}{c}
       \tilde{\mathbf{A}}_{i}\\
       g_{1,i}\tilde{\mathbf{A}}_{i}\\
       \vdots\\
       g_{1,i}^{k-2}\tilde{\mathbf{A}}_{i}
    \end{array}\right)=
    ((\tilde{H}_s\cdot f)(a_{i,1}),\ldots,(\tilde{H}_s\cdot f)(a_{i,r})).
\end{equation*}
On the other hand, by $3)$ in Lemma \ref{lem-a1}, we know that
\begin{equation}\label{eq3-2-9}
    \left(\begin{array}{cccc}
       \tilde{\mathbf{A}}_{1} & \tilde{\mathbf{A}}_{2} & \cdots & \tilde{\mathbf{A}}_{k}\\
       g_{1,1}\tilde{\mathbf{A}}_{1} & g_{1,2}\tilde{\mathbf{A}}_{2} & \cdots & g_{1,k}\tilde{\mathbf{A}}_{k}\\
       \vdots & \vdots &  & \vdots\\
       g_{1,1}^{k-2}\tilde{\mathbf{A}}_{1} & g_{1,2}^{k-2}\tilde{\mathbf{A}}_{2} & \cdots & g_{1,k}^{k-2}\tilde{\mathbf{A}}_{k}
    \end{array}\right)\cdot
    \left(\begin{array}{c}
       \frac{\tilde{\mathbf{A}}_{1}^{-1}}{h_{G_1\setminus\{g_{1,1}\}}(g_{1,1})} \\
       \frac{\tilde{\mathbf{A}}_{2}^{-1}}{h_{G_1\setminus\{g_{1,2}\}}(g_{1,2})}\\
       \vdots \\
       \frac{\tilde{\mathbf{A}}_{k}^{-1}}{h_{G_1\setminus\{g_{1,k}\}}(g_{1,k})} 
    \end{array}\right)
\end{equation}
is an all-$0$ matrix of size $kr\times r$. Then, by (\ref{eq3-2-9}), $\sum_{i\in [k]}
    ((\tilde{H}_s\cdot f)(a_{i,1}),\ldots,(\tilde{H}_s\cdot f)(a_{i,r}))\cdot\frac{\tilde{\mathbf{A}}_i^{-1}\cdot \mathbf{u}_i}{h_{G_1\setminus\{g_{1,i}\}}(g_{1,i})}$ equals to
\begin{equation*}
    \tilde{\mathbf{v}}\cdot\left(\begin{array}{ccc}
       \tilde{\mathbf{A}}_{1} & \cdots & \tilde{\mathbf{A}}_{k}\\
       g_{1,1}\tilde{\mathbf{A}}_{1} & \cdots & g_{1,k}\tilde{\mathbf{A}}_{k}\\
       \vdots &  & \vdots\\
       g_{1,1}^{k-2}\tilde{\mathbf{A}}_{1} & \cdots & g_{1,k}^{k-2}\tilde{\mathbf{A}}_{k}
    \end{array}\right)\cdot
    \left(\begin{array}{c}
       \frac{\tilde{\mathbf{A}}_{1}^{-1}}{h_{G_1\setminus\{g_{1,1}\}}(g_{1,1})} \\
       \vdots \\
       \frac{\tilde{\mathbf{A}}_{k}^{-1}}{h_{G_1\setminus\{g_{1,k}\}}(g_{1,k})} 
    \end{array}\right)\cdot 
    \left(\begin{array}{ccc}
       \mathbf{u}_1 &  & \\
       &  \ddots & \\
        &  &  \mathbf{u}_k
    \end{array}\right)=0.
\end{equation*}
By (\ref{eq3-2-7}) and (\ref{eq3-2-8}), this implies $\mathbf{v}_{H_C\cdot f}\cdot \mathbf{M}\cdot \mathbf{c}_{s,t}=0$ for every $f\in V_{k-l,r}$ and $(s,t)\in [l]\times[r]$. Note that 
$V_0=\{\mathbf{v}_{H_C\cdot f}: f\in V_{k-l,r}\}$ is a $(k-l)r$-dim subspace of $\mathbb{F}_q^{kr}$ and for every $(s,t)\in [l]\times[r]$
$$\mathbf{v}_{H_C\cdot f}\cdot \mathbf{c}_{s,t}=(H_C\cdot f)(c_{s,t})=0.$$
Thus, $\text{Span}_{\mathbb{F}_q}\{\mathbf{c}_{i,j}:~(i,j)\in [l]\times [r]\}$ is the dual space of $V_0$ in $\mathbb{F}_q^{kr}$. Then, $\mathbf{v}_{H_C\cdot f}\cdot \mathbf{M}\cdot \mathbf{c}_{s,t}=0$ implies that $\mathbf{M}\cdot \mathbf{c}_{s,t}\in \text{Span}_{\mathbb{F}_q}\{\mathbf{c}_{i,j}:~(i,j)\in [l]\times [r]\}$. This confirms the second property.
\end{proof}

In the following, using the multiplicative group of $\mathbb{F}_q$, we give an explicit construction of $A_{s,i}$, $C_i$ and $g(x)$ that satisfy the requirements in Corollary \ref{coro3-1}. 

Let $\zeta\geq 2$, $r$, $k$, $l^I$ and $l^F\leq \min\{k,l^I\}$ be positive integers. Let $q$ be a prime power such that $k(r+1)|(q-1)$ and
$$q\geq k(r+1)\cdot \max\{\zeta+1,\lceil\frac{l^I}{k}\rceil+2\}+1.$$ 
Let $\beta$ be the generator of the multiplicative group $\mathbb{F}_q^{*}$. Denote $\alpha=\beta^{\frac{q-1}{k(r+1)}}$ and let $G$ be the subgroup of $\mathbb{F}_q^{*}$ generated by $\alpha$. Denote $\zeta_0=\max\{\zeta,\lceil\frac{l^I}{k}\rceil+1\}$. Clearly, $|G|=k(r+1)$ and $|\mathbb{F}_q^{*}/G|\geq \zeta_0+1$.  For every $(s,i,j)\in [0,\zeta_0]\times [k]\times [r+1]$, let $a_{s,i,j}=\beta^s\alpha^{i-1+jk}$. Now, we define $g(x)=x^{r+1}$ and
\begin{itemize}
    \item $A_{s,i}=\{a_{s-1,i,1},\ldots,a_{s-1,i,r+1}\}$, $(s,i)\in [\zeta_0]\times [k]$ and $A_{s}=\bigcup_{i=1}^{k}A_{s,i}$;
    \item $C_i=\{a_{\zeta_0,i,1},\ldots, a_{\zeta_0,i,r+1}\}$, $1\leq i\leq l^F$ and $C=\bigcup_{i=1}^{l^F}C_{i}$.
\end{itemize}
Since $C_i\subseteq \beta^{\zeta_0}G$ and $A_{s,i}\subseteq \beta^{s-1} G$, $A_{s,i}$'s and $C_i$'s are mutually disjoint. Notice that $\zeta_0-1\geq \lceil\frac{l^I}{k}\rceil\geq \lceil\frac{l^I-l^F}{k}\rceil$. Thus, $(\zeta_0-1)k\geq l^I-l^F$ and we can pick $l^I-l^F$ different $(r+1)$-subsets from $\{A_{s,i}\}_{2\leq s\leq \zeta_0\atop 1\leq i\leq k}$ as $B_1,\ldots,B_{l^I-l^F}$. Denote $B=\bigcup_{i=1}^{l^I-l^F}B_i$. Clearly, $B_1,\ldots,B_{l^I-l^F}$ are mutually disjoint and $B\cap (A_1\cup C)=\emptyset$.

Next, we show that the $A_{s,i}$, $C_i$ and $g$ defined above satisfy the requirements in Corollary \ref{coro3-1}:
\begin{itemize}
    \item For each $(s,i,j)\in [0,\zeta_0]\times [k]\times [r+1]$, by $\alpha^{k(r+1)}=1$, we have
    \begin{align*}
       g(\beta^{s}\alpha^{i-1+jk})&=\left(\beta^{s}\alpha^{i-1+jk}\right)^{r+1}\\
       &=\left(\beta^{s}\alpha^{i-1}\right)^{r+1}.
    \end{align*}
    Thus, $g$ is constant on each $A_{s,i}$ and $C_i$. 
    \item Since $A_{s+1,i}=\{\beta^{s}\alpha^{i-1+k},\ldots, \beta^{s}\alpha^{i-1+(r+1)k}\}$, we have
    \begin{align*}
       \mathbf{A}_{s+1,i}^{(r)}&=\left(\begin{array}{cccc}
       1  & 1 & \cdots & 1 \\
       \beta^{s}\alpha^{i-1+k} & \beta^{s}\alpha^{i-1+2k} & \cdots & \beta^{s}\alpha^{i-1+(r+1)k} \\
       \vdots & \vdots &  & \vdots \\
       (\beta^{s}\alpha^{i-1+k})^{r-1} & (\beta^{s}\alpha^{i-1+2k})^{r-1} & \cdots & (\beta^{s}\alpha^{i-1+(r+1)k})^{r-1}
    \end{array}\right) \nonumber \\
    &=\text{diag}(1,\beta^{s-1},\ldots, \beta^{(s-1)(r-1)})\mathbf{A}_{1,i}^{(r)}.
    \end{align*}
    Take $\mathbf{T}_s=\text{diag}(1,\beta^{s-1},\ldots, \beta^{(s-1)(r-1)})$, clearly, $\mathbf{T}_s$ is invertible. This verifies the condition i).
    \item Let $g_{s,i}$ be the constant $g(A_{s,i})$. Note that for every $1\leq s\leq \zeta$,
    \begin{align*}
       G_s=\{g_{s,1},\ldots, g_{s,k}\}=\{\beta^{(s-1)(r+1)},\ldots,(\beta^{(s-1)}\alpha^{k-1})^{r+1}\}.
    \end{align*}
    Thus, we have
    \begin{align*}
       \mathbf{G}_{s}^{(l^F)}&=\left(\begin{array}{cccc}
       1  & 1 & \cdots & 1 \\
       \beta^{(s-1)(r+1)} & (\beta^{(s-1)}\alpha)^{r+1} & \cdots & (\beta^{(s-1)}\alpha^{k-1})^{r+1} \\
       \vdots & \vdots &  & \vdots \\
       \beta^{(s-1)(r+1)(l^F-1)} & (\beta^{(s-1)}\alpha)^{(r+1)(l^F-1)} & \cdots & (\beta^{(s-1)}\alpha^{k-1})^{(r+1)(l^F-1)}
    \end{array}\right) \nonumber \\
    &=\text{diag}(1,\beta^{(s-1)(r+1)},\ldots, \beta^{(s-1)(r+1)(l^F-1)})\mathbf{G}_{1}^{(l^F)}.
    \end{align*}
    Take $\mathbf{T}_s'=\text{diag}(1,\beta^{(s-1)(r+1)},\ldots, \beta^{(s-1)(r+1)(l^F-1)})$, clearly, $\mathbf{T}_s'$ is invertible. This verifies the condition ii).
\end{itemize}

\begin{remark}
The sets $A_1$, $A_2$, $B_1$, $B_2$, and $C$ in Example \ref{ex2} can be viewed as an example of the above construction for $q=19$, $k=r=2$, $l^I=l^F=1$ and $\zeta=2$.
\end{remark}

Clearly, the sets $A_{s,i}$'s and $C_i$ constructed above satisfy the requirements in Corollary \ref{coro3-1}. Then, $A_{s,i}$, $B_i$ and $C_i$ constructed above satisfy the requirements in Construction IV. Thus, we have the following immediate corollary.

\begin{corollary}\label{coro3}
For positive integers $\zeta\geq 2$, $r$, $k$, $l^I$ and $l^F\leq \min\{k,l^I\}$, let $q$ be a prime power such that $k(r+1)\mid (q-1)$ and $q\geq k(r+1)\cdot \max\{\zeta+1,\lceil\frac{l^I}{k}\rceil+2\}+1$. Then, there is an explicit construction of $(n^I,kr,r;n^F,\zeta kr,r)$ optimal LRCCs with write access cost $l(r+1)$ and read access cost $\zeta lr$, where $n^I=(k+l^I)(r+1)$ and $n^F=(k+l^F)(r+1)$.
\end{corollary}

\section{Lower bounds on the access cost}\label{sec_lb}

In this section, we prove a general lower bound on the access cost of an $(n^I,k;n^F,\zeta k)$ convertible code with the condition that $\mathcal{C}^F$ is an $r$-LRC. As a consequence, the LRCCs obtained by Constructions III and IV admit optimal access costs.

Throughout the section, we need the following additional notations. Let $\mathcal{C}$ be an $(n,k)$ code over $\mathbb{F}_q$ and $\mathbf{m}\in \mathbb{F}_q^{k}$ be a message vector, denote $\mathcal{C}(\mathbf{m})$ as the codeword in $\mathcal{C}$ that encodes $\mathbf{m}$. Let $(\mathcal{C}^I,\mathcal{C}^{F})$ be an $(n^{I},k;n^{F},\zeta k)$ convertible code over $\mathbb{F}_q$ with conversion procedure $T$.
For a message vector $\mathbf{m}\in\mathbb{F}_q^{\zeta k}$ and $1\leq i\leq \zeta$, we denote $\mathbf{c}_i=\mathcal{C}^I(\mathbf{m}|_{[(i-1)k+1,ik]})$ as the $i$-th initial codeword and denote $\mathbf{d}=\mathcal{C}^{F}(\mathbf{m})$ as the final codeword. For $i\in [\zeta]$, we denote $U_i\subseteq [n^{I}]$ as the set of coordinates corresponding to the remaining symbols in $\mathbf{c}_i$ and $A_i\subseteq [n^{I}]$ as the set of coordinates corresponding to the accessed symbols in $\mathbf{c}_i$. We denote $T(U_i)\subseteq [n^{F}]$ as the set of coordinates of the symbols in $\mathbf{d}$ that remains from $\mathbf{c}_i$. Moreover, for $j\in T(U_i)$, we say the symbol $\mathbf{d}(j)$ is inherited from $\mathbf{c}_i|_{U_i}$. Clearly, there is a one-to-one map between $U_i$ and $T(U_i)$ and for convenience we also use $T$ to denote this map. In the following, by a slight abuse of notation, we also use the coordinate $j$ to referred to the corresponding codeword symbol.

The proof of the general lower bound consists of two parts. First, we prove a lower bound on the write access cost. To do this, we bound from above the number of remaining symbols in each initial codewords. Then, we prove a lower bound on the read access cost based on the locality of $\mathcal{C}^{F}$.

To start with, we prove the following upper bound on $|U_i|$.

\begin{lemma}\label{lem1}
Let $(\mathcal{C}^I,\mathcal{C}^F)$ be an $(n^{I},k;n^{F},\zeta k)$ convertible code over $\mathbb{F}_q$ with conversion procedure $T$. Assume that $\mathcal{C}^F$ has locality $r$ and minimum distance $d$. Then, for every $1\leq i \leq \zeta$, we have
\begin{equation}\label{upb_R}
   |U_i|=|T(U_i)|\leq n^{F}-d-((\zeta-1) k+\lceil\frac{(\zeta-1) k}{r}\rceil)+2.
\end{equation}
\end{lemma}
\begin{proof}
We only show that (\ref{upb_R}) holds for $i=1$ and the proofs for cases $2\leq i\leq \zeta$ are similar.

Let $\tilde{\mathcal{C}}$ be the subcode of $\mathcal{C}^{F}$ that encodes message vectors in $\{\mathbf{m}\in\mathbb{F}_q^{\zeta k}: \mathbf{m}|_{[k]}=\mathbf{0}\}$, i.e.,
$$\tilde{\mathcal{C}}=\{\mathbf{d}\in \mathcal{C}^F: \mathbf{d} =\mathcal{C}^{F}((\mathbf{0},\tilde{\mathbf{m}})) \text{~for some~} \tilde{\mathbf{m}}\in \mathbb{F}_q^{(\zeta-1)k}\}.$$
Denote $d(\tilde{\mathcal{C}})$ as the minimal distance of $\tilde{\mathcal{C}}$. Clearly, we have $d(\tilde{\mathcal{C}})\geq d$ and $|\tilde{\mathcal{C}}|=q^{(\zeta-1)k}$.

Define
$$N=\{i\in [n^F]: ~i\in T(U_1),~\text{or}~i~\text{can be recovered by symbols in}~T(U_1)\}.$$
Denote $N^{C}=[n^F]\setminus N$. Then, for every $i\in N^C$ and every recovering set $I_i$ of $i$, we have $I_i\cap N^C\neq \emptyset$. Otherwise, assume that there is an $i_0\in N^{C}$ with a recovering set $I_0\subseteq N$. Then, by the definition of $N$, every symbol of $I_0$ can be recovered by symbols in $T(U_1)$. Thus, $i_0$ can also be recovered by symbols in $T(U_1)$. This leads to $i_0\in N$, a contradiction.

Let $\mathbf{v}_0=\mathcal{C}^I(\mathbf{0})$ and $\mathbf{d}=\mathcal{C}^{F}((\mathbf{0},\tilde{\mathbf{m}}))\in \tilde{\mathcal{C}} \text{~for some~} \tilde{\mathbf{m}}\in \mathbb{F}_q^{(\zeta-1)k}$. Then, $\mathbf{d}$ is converted from $\mathbf{c}_1=\mathbf{v}_0$ and $\mathbf{c}_{i+1}=\mathcal{C}^I(\tilde{\mathbf{m}}|_{[(i-1)k+1,ik]})$, $1\leq i\leq \zeta-1$. Note that each symbol of $\mathbf{d}|_{T(U_1)}$ is from $\mathbf{v}_0|_{U_1}$ and symbols in $\mathbf{d}|_{N}$ can be recovered by symbols in $\mathbf{d}|_{T(U_1)}$. Thus, $\mathbf{d}|_{N}$ is fully determined by $\mathbf{v}_0|_{U_1}$ regardless of $\tilde{\mathbf{m}}$. This implies that $\tilde{\mathcal{C}}|_{N}$ is a constant vector determined by $\mathbf{v}_0|_{U_1}$. Therefore, $|\tilde{\mathcal{C}}|=|\tilde{\mathcal{C}}|_{N^{C}}|=q^{(\zeta-1)k^I}$.

On the other hand, note that each $i\in N^C$ has a recovering set $I_i$ such that $I_i\cap N^{C}\neq\emptyset$. Since $\tilde{\mathcal{C}}|_{N}$ is a constant vector, we have $d(\tilde{\mathcal{C}})=d(\tilde{\mathcal{C}}|_{N^C})$. Moreover, $\tilde{\mathcal{C}}|_{N^C}$ can be viewed as an $(n^F-|N|,(\zeta-1)k,r)$-LRC by taking $I_i\cap N^C$ as the recovering set for every $i\in N^C$. Thus, by Theorem \ref{Singleton_bound}, we have
\begin{equation}\label{eq02}
d(\tilde{\mathcal{C}})=d(\tilde{\mathcal{C}}|_{N^C})\leq (n^{F}-|N|)-((\zeta-1) k+\lceil\frac{(\zeta-1) k}{r}\rceil)+2.
\end{equation}
Then, the result follows directly from $d(\tilde{\mathcal{C}})\geq d$, $|U_1|\leq |N|$ and (\ref{eq02}).
\end{proof}

Now, we focus on bounding the read access cost $|A_i|$, i.e., the number of symbols that are read from the $i$-th initial codeword during the conversion. 

\begin{lemma}\label{lem2}
Let $(\mathcal{C}^I,\mathcal{C}^F)$ be the convertible code defined in Lemma \ref{lem1}.
For each $1\leq i\leq \zeta$, denote $\Delta_i= |U_i\setminus A_i|-d+1$. Let $B$ be the largest subset of $T(U_i\setminus A_i)$ with $|B|\leq \lfloor\frac{\Delta_i}{r+1}\rfloor$ such that every $b\in B$ has a recovering set $I_b$ satisfying $I_b\cap B=\emptyset$.
Then, for each $1\leq i\leq \zeta$, we have
\begin{equation}\label{lb_D}
    |A_i|\geq
    \begin{cases}
    k,~\text{if $\Delta_i\leq 0$};\\
    k-\Delta_i+|B|,~\text{otherwise}.
    \end{cases}
\end{equation}
Moreover, if $d> n^I-k+1$, then $\Delta_i\leq 0$ for every $1\leq i\leq \zeta$ and therefore, $|A_i|\geq k$.
\end{lemma}
\begin{proof}
Similar to the proof of Lemma \ref{lem1}, we only show that (\ref{lb_D}) holds for $i=1$.

Let $\tilde{\mathcal{C}}$ be the subcode of $\mathcal{C}^{F}$ that encodes message vectors in $\{\mathbf{m}\in\mathbb{F}_q^{\zeta k}: \mathbf{m}|_{[k+1,\zeta k]}=\mathbf{0}\}$, i.e.,
$$\tilde{\mathcal{C}}=\{\mathbf{d}\in \mathcal{C}^F: \mathbf{d}=\mathcal{C}^{F}((\tilde{\mathbf{m}},\mathbf{0})) \text{~for some~} \tilde{\mathbf{m}}\in \mathbb{F}_q^{k}\}.$$
Denote $d(\tilde{\mathcal{C}})$ as the minimal distance of $\tilde{\mathcal{C}}$. Clearly, $d(\tilde{\mathcal{C}})\geq d$ and $|\tilde{\mathcal{C}}|=q^{k}$.

Now, fix a vector $\tilde{\mathbf{m}}\in\mathbb{F}_q^{k}$, let $\mathbf{d}=\mathcal{C}^{F}((\tilde{\mathbf{m}},\mathbf{0}))\in \tilde{\mathcal{C}}$. Each codeword symbol of $\mathbf{d}$ is either inherited from $\mathbf{c}_1|_{U_1},\ldots,\mathbf{c}_{\zeta}|_{U_{\zeta}}$, or is determined by symbols from $\mathbf{c}_{1}|_{A_1},\ldots,\mathbf{c}_{\zeta}|_{A_{\zeta}}$. Note that for $2\leq i\leq \zeta$,
$$\mathbf{c}_i=\mathcal{C}^{I}((\tilde{\mathbf{m}},\mathbf{0})|_{[(i-1)k+1,ik]})=\mathcal{C}^{I}(\mathbf{0})$$ 
is a constant vector. Therefore, the non-constant symbols of $\mathbf{d}$ are either inherited from $\mathbf{c}_1|_{U_1\setminus A_1}$, or are functions of symbols in $\mathbf{c}_{1}|_{A_1}$.

For any $(n,k)$ code $\mathcal{C}$ with distance $d$, we know that $|\mathcal{C}|_{S}|=q^{k}$ for any subset $S\subseteq [n]$ of size at least $n-d+1$. Based on this observation, we will pick a subset $S\subseteq [n^F]$ of size at least $n^F-d+1$ and prove the lower bound by analyzing the relationships among codeword symbols in $\mathbf{c}_1|_{U_1\setminus A_1}$, $\mathbf{c}_{1}|_{A_1}$ and $\mathbf{d}|_{S}$.

\begin{itemize}
  \item \textbf{Case 1}: $\Delta_1\leq 0$
\end{itemize}

Recall that there is a one-to-one map between $U_1$ and $T(U_1)$, thus, $\Delta_1=|T(U_1\setminus A_1)|-d+1\leq 0$. In other words, $[n^F]\setminus T(U_1\setminus A_1)$ has size at least $n^F-d+1\geq n^F-d(\tilde{\mathcal{C}})+1$. Take $S=[n^F]\setminus T(U_1\setminus A_1)$, we have $|\tilde{\mathcal{C}}|_S|=q^{k}$. Note that symbols in $\mathbf{c}_{1}|_{U_1\setminus A_1}$ are remaining as $\mathbf{d}|_{T(U_1\setminus A_1)}$ in $\mathbf{d}$. Thus, by $S\cap T(U_1\setminus A_1)=\emptyset$, non-constant symbols in $\mathbf{d}|_{S}$ are functions of symbols in $\mathbf{c}_{1}|_{A_1}$. Therefore, $|\tilde{\mathcal{C}}|_S|$ is upper bounded by $q^{|A_1|}$, which implies that $|A_1|\geq k$.

\begin{itemize}
  \item \textbf{Case 2}: $\Delta_1> 0$
\end{itemize}

Let $N=\bigcup_{b\in B}I_b$. Since $\mathcal{C}^{F}$ has locality $r$, we have $|B\cup N|\leq |B|(r+1)\leq \Delta_1$. 
Note that
\begin{align*}
    |T(U_1\setminus A_1)\setminus (B\cup N)|&\geq |T(U_1\setminus A_1)|-|B\cup N|\\
    &\geq |T(U_1\setminus A_1)|-\Delta_1=d-1.
\end{align*}
Therefore, we can obtain an $(n^{F}-d+1)$-subset $S\subseteq[n^F]$ by removing any $d-1$ elements in $T(U_1\setminus A_1)\setminus (B\cup N)$ from $[n^F]$. Clearly, both $B\cup N$ and $[n^F]\setminus T(U_1\setminus A_1)$ are subsets of $S$. By $n^{F}-d+1\geq n^{F}-d(\tilde{\mathcal{C}})+1$, we have $|\tilde{\mathcal{C}}|_S|=q^{k}$. 

Let $S_1=S\cap T(U_1\setminus A_1)$. Then, 
$$|S_1|=|T(U_1\setminus A_1)|-d+1=\Delta_1.$$
Since $B\subseteq T(U_1\setminus A_1)$ and $B\cup N\subseteq S$, we have $B\subseteq S_1$. Moreover, by
$N\cap B=\emptyset$, we have $N\subseteq S\setminus B$. Thus, symbols in $\mathbf{d}|_B$ can be recovered by symbols in $\mathbf{d}|_{S\setminus B}$. On the other hand, note that symbols in $\mathbf{c}_{1}|_{U_1\setminus A_1}$ are remaining as $\mathbf{d}|_{T(U_1\setminus A_1)}$. Thus, by $T(U_1\setminus A_1)\cap S=S_1$, symbols in $\mathbf{d}|_{S}$ that are inherited from $\mathbf{c}_{1}|_{U_1\setminus A_1}$ are from $\mathbf{c}_{1}|_{T^{-1}(S_1)}$. This implies that non-constant symbols in $\mathbf{d}|_{S\setminus B}$ is determined by symbols in $\mathbf{c}_{1}|_{A_1}$ and $\mathbf{c}_{1}|_{T^{-1}(S_1\setminus B)}$. Thus, $\mathbf{d}|_S$ is determined by $\mathbf{c}_{1}|_{A_1}$ and $\mathbf{c}_{1}|_{T^{-1}(S_1\setminus B)}$. Therefore, $|\tilde{\mathcal{C}}|_S|$ is upper bounded by
$$q^{|A_1\cup T^{-1}(S_1\setminus B)|}=q^{|A_1|+|S_1|-|B|},$$
where the equality follows from $T^{-1}(S_1\setminus B)\subseteq U_1\setminus A_1$ and $(U_1\setminus A_1)\cap A_1=\emptyset$.
This leads to $|A_1|+|S_1|-|B|\geq k$, which further implies that $|A_1|\geq k-\Delta_1+|B|$.

\begin{itemize}
  \item \textbf{Case 3}: $d> n^I-k+1$
\end{itemize}

In this case, we claim that one must have $\Delta_1\leq 0$. Otherwise, assume that $\Delta_1>0$ and let $C\subseteq U_1\setminus A_1$ such that $|C|=d-1$. Let $S=[n^F]\setminus T(C)$. By $|S|=n^F-d+1\geq n^F-d(\tilde{\mathcal{C}})+1$, we have $|\tilde{\mathcal{C}}|_S|=q^{k}$. Note that for each $\mathbf{d}\in \tilde{\mathcal{C}}$, the non-constant symbols of $\mathbf{d}|_S$ are either inherited from $\mathbf{c}_{1}|_{U_1\setminus C}$ or determined by $\mathbf{c}_{1}|_{A_1}$. Therefore, we have $|\tilde{\mathcal{C}}|_S|=q^{k}\leq q^{|(U_1\setminus C)\cup A_1|}$. Since $C\subseteq U_1\setminus A_1$, we have $(U_1\setminus C)\cup A_1=(U_1\cup A_1)\setminus C$ and $|(U_1\cup A_1)\setminus C|=|U_1\cup A_1|-|C|$. This leads to
$$|(U_1\setminus C)\cup A_1|=|U_1\cup A_1|-|C|\geq k.$$
On the other hand, since $U_1\cup A_1\subseteq [n^I]$, we have
$$|U_1\cup A_1|-|C|\leq n^I-d+1.$$
This implies that $d\leq n^I-k+1$, a contradiction. Thus, when $d>n^I-k+1$, we always have $\Delta_1\leq 0$ and therefore, $|A_1|\geq k$.

Combining the above cases together, we can conclude the result.
\end{proof}

In Lemma \ref{lem2}, the size of the special subset $B$ affects the lower bound on the read access cost $|A_i|$. To get an exact lower bound on $|A_i|$, in the following, we show that one can always find such $B$ with $|B|=\lfloor\frac{\Delta_i}{r+1}\rfloor$ when $\mathcal{C}^{F}$ is linear.


Let $\mathcal{C}$ be an $(n,k,r)$-LRC. We call a subset $R_i\subseteq [n]$ an \emph{extended recovering set} of $i$, if $i\in R_i$ and $R_i\setminus\{i\}$ is a recovering set of $i$.



\begin{lemma}\label{lem3}
Assume that $\mathcal{C}^F$ is a linear $(n^F,\zeta k,r)$-LRC. For any $A\subseteq [n^F]$, there is a subset $A'\subseteq A$ of size $\lceil\frac{|A|}{r+1}\rceil$ such that every $i\in A'$ has a recovering set $I_i$ satisfying $I_i\cap A'=\emptyset$. 
\end{lemma}

\begin{proof}
Let $\Gamma\subseteq 2^{[n^F]}$ be the family of all possible extended recovering sets of $\mathcal{C}^{F}$. First, we show that there is a sub-family $\mathcal{R}\subseteq \Gamma$ such that $\bigcup_{R\in \mathcal{R}}R=[n^F]$ and for any $R'\in \mathcal{R}$, $\bigcup_{R\in \mathcal{R}\setminus\{R'\}}R\neq [n^F]$. Clearly, $\bigcup_{R\in \Gamma}R=[n^F]$. Assume that there is an $R_0\in \Gamma$ such that $\bigcup_{R\in \Gamma\setminus\{R_0\}}R= [n^F]$. Then, we can remove this $R_0$ from $\Gamma$ and obtain a new family $\Gamma_1$ such that $\bigcup_{R\in \Gamma_1}R=[n^F]$. Since $\Gamma$ is finite, by repeating this procedure, we can obtain a sub-family $\mathcal{R}\subseteq \Gamma$ satisfy the required property at a certain step.

Denote 
$$\mathcal{S}=\{R\cap A: R\in \mathcal{R}\}.$$
Clearly, $\bigcup_{S\in \mathcal{S}}S=A$. W.l.o.g., we assume that for every $S'\in \mathcal{S}$, $\bigcup_{S\in \mathcal{S}\setminus \{S'\}}S\neq A$. Otherwise, by apply the elimination procedure above, we can get a sub-family $\mathcal{S}'\subseteq \mathcal{S}$ satisfying this property and $\bigcup_{S\in \mathcal{S}'}S=A$. 

Under this assumption, each $S\in \mathcal{S}$ contains a unique element $i_S\in A$ which doesn't belong to $\bigcup_{S'\in \mathcal{S}\setminus\{S\}}S'$. By the linearity of $\mathcal{C}$, for every $R\in\mathcal{R}$ and $i\in R$, $R\setminus \{i\}$ is a recovering set of $i$. Therefore, for $S\in \mathcal{S}$, we can take some $R(S)\in \mathcal{R}$ satisfying $R(S)\cap A=S$ as the extended recovering set of $i_{S}$. Clearly, for any $S'\in \mathcal{S}\setminus \{S\}$, $i_{S'}\notin R(S)$. Let $A'=\{i_{S}: S\in \mathcal{S}\}$. Then, we have $|A'|=|\mathcal{S}|$ and every $i_{S}\in A'$ has a recovering set $R(S)\setminus\{i_{S}\}$ such that $R(S)\cap A'=\{i_{S}\}$. Since $|R|\leq r+1$ for each $R\in\mathcal{R}$, $|S|\leq r+1$ for each $S\in \mathcal{S}$. Therefore, $|A'|=|\mathcal{S}|\geq \lceil\frac{|A|}{r+1}\rceil$. 
\end{proof}

Now, armed with the results above, we can obtain the following general lower bound on the access cost of the convertible code $(\mathcal{C}^{I},\mathcal{C}^{F})$.

\begin{theorem}\label{lb_accesscost}
Let $(\mathcal{C}^I,\mathcal{C}^F)$ be an $(n^{I},k;n^{F},\zeta k)$ convertible code over $\mathbb{F}_q$ with conversion procedure $T$. Assume that $\mathcal{C}^F$ is a linear code with locality $r$ and minimum distance $d$. Then, the write access cost of $T$ is at least
\begin{equation}\label{eq_main1}
  \zeta \left(d+ (\zeta-1)k+\lceil\frac{(\zeta-1) k}{r}\rceil-2\right)-(\zeta-1)n^F.
\end{equation}
Write $\Delta=n^{F}-2d-((\zeta-1) k+\lceil\frac{(\zeta-1) k}{r}\rceil)+3$. Then, the read access cost of $T$ is at least
\begin{equation}\label{eq_main2}
\begin{cases}
\zeta k, \text{ if $\Delta\leq0$};\\
\zeta(k-\lceil\frac{r\Delta}{r+1}\rceil), \text{ otherwise.}
\end{cases}
\end{equation}
Moreover, when $d> n^I-k+1$, the read access cost is at least $\zeta k$.
\end{theorem}

\begin{proof}
Note that the write access cost is the number of new symbols written during the conversion procedure and there are $n^F$ symbols in total. Therefore, (\ref{eq_main1}) follows directly from Lemma \ref{lem1}.

As in Lemma \ref{lem2}, for every $i\in [\zeta]$, assume that $|U_i\setminus A_i|-d+1=\Delta_i$. By Lemma \ref{lem1}, we have 
$$|U_i\setminus A_i|\leq |U_i|\leq n^{F}-d-((\zeta-1) k+\lceil\frac{(\zeta-1) k}{r}\rceil)+2.$$
This leads to,
\begin{align*}
  |U_i\setminus A_i|-d+1 & \leq  n^{F}-2d-((\zeta-1) k+\lceil\frac{(\zeta-1) k}{r}\rceil)+3=\Delta.
\end{align*}
Thus, we have $\Delta_i\leq \Delta$ for every $1\leq i\leq \zeta$. On the other hand, by Lemma \ref{lem3}, we have
\begin{align*}
\Delta_i-|B|&=\Delta_i-\lfloor\frac{\Delta_i}{r+1}\rfloor\\
&=\lceil\Delta_i\cdot\frac{r}{r+1}\rceil\leq \lceil\frac{r\Delta}{r+1}\rceil.
\end{align*} 
Therefore, by Lemma \ref{lem2}, we have
\begin{equation*}
    |A_i|\geq
    \begin{cases}
    k,~\text{if $\Delta\leq 0$};\\
    k-\lceil\frac{r\Delta}{r+1}\rceil,~\text{otherwise}.
    \end{cases}
\end{equation*}
for every $i\in [\zeta]$. This leads to (\ref{eq_main2}).

When $d> n^I-k+1$, the results follow directly from Lemma \ref{lem2}.
\end{proof}

\begin{corollary}\label{coro1}
Let $(\mathcal{C}^I,\mathcal{C}^F)$ be an $(n^{I},kr,r;n^{F},\zeta kr,r)$ optimal LRCCs over $\mathbb{F}_q$ with conversion procedure $T$. Assume that $(r+1)|n^I$ and $(r+1)|n^F$. Then, the write access cost of $T$ is at least $n^F-\zeta k(r+1)$ and the read access cost of $T$ is at least
\begin{equation}\label{eq_main4}
    \begin{cases}
    \zeta kr, \text{ if $(\zeta+1)k\leq \frac{n^F}{r+1}$};\\
    \zeta(\frac{rn^F}{r+1}-\zeta kr), \text{ otherwise.}
    \end{cases}
\end{equation}
Moreover, when $n^F-n^I\geq (\zeta-1)kr+\zeta k$, the read access cost is at least $\zeta kr$.
\end{corollary}

\begin{remark}\label{rmk1}
The LRCC obtained in Corollary \ref{coro3} attains the lower bound in (\ref{eq_main4}).
\end{remark}

\begin{proof}
The proof is straight forward. Since $\mathcal{C}^F$ is an optimal $r$-LRC, we have $d=d(\mathcal{C}^F)=n^F-\zeta k (r+1)+2$. Thus, by (\ref{eq_main1}), the write access cost is at least $n^F-\zeta k(r+1)$. Meanwhile, by the formula of $\Delta$, we have
\begin{align*}
    \Delta&=n^{F}-2d-((\zeta-1) kr+\lceil\frac{(\zeta-1) kr}{r}\rceil)+3\\
    &=2\zeta k(r+1)-4-n^F-(\zeta-1) k(r+1)+3\\
    &=(\zeta+1)k(r+1)-n^F-1.
\end{align*}
Thus, $\lfloor\frac{\Delta}{r+1}\rfloor=(\zeta+1)k-\frac{n^F}{r+1}-1$ and $\Delta-\lfloor\frac{\Delta}{r+1}\rfloor=(\zeta+1)kr-\frac{rn^F}{r+1}$. This leads to
\begin{align*}
    kr-\lceil\frac{r\Delta}{r+1}\rceil&=kr-\left(\Delta-\lfloor\frac{\Delta}{r+1}\rfloor\right)\\
    &=\frac{rn^F}{r+1}-\zeta kr
\end{align*}
and (\ref{eq_main4}) follows directly from (\ref{eq_main2}).

When $n^F-n^I\geq (\zeta-1)kr+\zeta k$, by $d=n^F-\zeta k(r+1)+2$, we have $d\geq n^I-kr+2$. Then, the result follows directly from Theorem \ref{lb_accesscost}.
\end{proof}

Note that $(n,k)$ MDS codes are $(n,k,k)$-optimal LRCs. Thus, by Lemma \ref{lem1} and Lemma \ref{lem2}, we can also obtain a proof of Theorem \ref{thm1} without assuming the linearity of $\mathcal{C}^I$ and $\mathcal{C}^F$.

\begin{theorem}\label{ac_MDS}
For all $(n^I,k;n^F,\zeta k)$ MDS convertible codes, the read access cost of conversion is at least $\zeta\cdot \min\{k,n^F-\zeta k\}$ and the write access cost is at least $n^F-\zeta k$. Furthermore, if $n^I-k<n^F-\zeta k$, then, the read access cost is at least $\zeta k$.
\end{theorem}

\begin{proof}
The proof is straight forward. Since $\mathcal{C}^F$ has locality $r=\zeta k$, we have $\lceil\frac{(\zeta-1) k}{r}\rceil=1$ and the minimum distance of $\mathcal{C}^F$ satisfies $d=n^F-\zeta k+1$. By substituting these two parameters into (\ref{upb_R}), we have $|U_i|\leq k$ for every $1\leq i\leq \zeta$. Thus, the write access cost is at least $n^F-\zeta k$. Meanwhile, by the definition of $\Delta_i$ in Lemma \ref{lem2}, we have
$$\Delta_i\leq |U_i|-d+1\leq(\zeta+1)k-n^F.$$
When $n^F-\zeta k\geq k$, we have $(\zeta+1)k-n^F\leq 0$. Thus, $\Delta_i\leq 0$ and by (\ref{lb_D}), $|A_i|\geq k$. When $n^F-\zeta k<k$, we have $0<(\zeta+1)k-n^F\leq k$. Thus, if $\Delta_i>0$, the requirement of $|B|$ in Lemma \ref{lem2} becomes $|B|\leq \lfloor\frac{\Delta_i}{r+1}\rfloor=0$. 
Then, we can take $B=\emptyset$ in Lemma \ref{lem2} and by (\ref{lb_D}), we have $|A_i|\geq k-\Delta_i\geq n^F-\zeta k$. Therefore, the read access cost is at least $\zeta\cdot \min\{k,n^F-\zeta k\}$. 

Moreover, if $n^I-k<n^F-\zeta k$, we have $d> n^I-k+1$ and the result also follows directly from Lemma \ref{lem2}.
\end{proof}

\section{Conclusion and further research}\label{sec_con}

In this paper, we study the code conversion problem of locally repairable codes in the merge regime. First, by establishing proper maps among polynomials of different degrees, we provide a general construction of MDS convertible codes with optimal access costs using GRS codes. As a consequence, we obtain a family of $(n^I,k;n^F,\zeta k)$ MDS convertible codes with optimal access costs over finite field of size linear in $n^I$. Then, based on the construction of optimal LRCs in \cite{TB14}, we extend the construction of MDS convertible codes and provide a general construction of LRCCs. Similarly, a family of $(n^I,k,r;n^F,\zeta k,r)$ LRCCs over finite field of size linear in $n^I$ is obtained. Finally, we prove a lower bound on the access cost of convertible codes with the condition that the final code is an LRC. This confirms that the LRCCs obtained through our construction admit optimal access costs.

The study presented in this paper is inspired by interesting open questions raised in \cite{MR22b} regarding constructions of low-field-size convertibles codes optimizing the conversion costs. Hence, our primary focus is on constructing low-field-size convertibles codes. As phrased in \cite{MR22b}, convertible codes have significant potentials for real-world impact by enabling resource-efficient redundancy tuning. Although many works have been done recently in this topic, there are still a lot of questions left. In the following, we mention several open questions that, in our opinion, could improve the study of locally repairable convertible codes.
\begin{enumerate}    
    \item[1.] In this paper, we only focus on LRCCs in the merge regime. However, as another fundamental case, convertible codes in split regime also have strong practical backgrounds. So, for LRCCs in split regime, what is the lower bound on the access cost and are there any corresponding optimal constructions? Very rescently, Maturana and Rashmi  \cite{MR23b} presented a construction of convertible codes with information locality in the split regime, which is a good start for the study of this question.
    \item[2.] For general LRCCs, the initial code and the final code may have different localities. However, our construction in Section \ref{sec_LRCC} requires the initial code and the final code have the same locality. Note that our lower bound in Theorem \ref{lb_accesscost} is for the general case. Is it possible to generalize our construction and obtain convertible codes with optimal access costs when the initial code and the final code have different localities?
    \item[3.] Unlike MDS codes, for optimal locally repairable codes, it is known that the size of the field can be sub-linear in the code length (e.g., see \cite{GXY19,CMST20,KWG21,XC22}). Note that long codes over small fields are preferred in practice. It's also interesting to ask how small the field size can be to ensure the existence of LRCCs with optimal access costs.
    \item[4.] As two variations of LRCs, codes with $(r,\delta)$-locality and LRCs with disjoint recovering sets (availability) are also important storage codes, which have attracted lots of attentions these years. Therefore, the study of code conversion problems for these codes are also interesting both theoretically and practically.
\end{enumerate}

\section*{Acknowledgement}

We thank Prof. Itzhak Tamo for his many illuminating discussions and especially, his helpful suggestions regarding the presentation of this paper. We also thank Prof. Xin Wang for reading and commenting the
earlier version of the manuscript.

\bibliographystyle{IEEEtran}
\bibliography{biblio}

\appendices

\section{Properties of $\{x^sg^{t}(x)\}_{0\leq s\leq r-1 \atop 0\leq t\leq k-1}$}

Let $k$ and $r$ be positive integers. For $1\leq i\leq k$, let $A_i=\{a_{i,1},\ldots,a_{i,r+1}\}\subset\mathbb{F}_q$ such that $A_i\cap A_i'=\emptyset$ for $i\neq i'$. Denote $A=\bigcup_{i=1}^{k}A_i$ and $\tilde{A}_i=A_i\setminus\{a_{i,r+1}\}$. Let $g(x)\in\mathbb{F}_q[x]$ be a polynomial of degree $r+1$ such that $g$ is constant on each $A_i$. Denote $g_i$ as the constant $g(A_i)$ and denote $G=\{g_1,\ldots,g_k\}$. For every $(i,j)\in [k]\times [r+1]$, define $\mathbf{a}_{i,j}\in\mathbb{F}_q^{kr}$ as  
$$\mathbf{a}_{i,j}=(a_{i,j}^{s}g^{t}(a_{i,j}),~(s,t)\in [0,r-1]\times [0,k-1])^{T}.$$

\begin{lemma}\label{lem-a1}
\begin{itemize}
    \item [1)] $\{\mathbf{a}_{i,j}\}_{1\leq i\leq k \atop 1\leq j\leq r}$ are linearly independent over $\mathbb{F}_q^{kr}$.
    \item [2)] For $c\in \mathbb{F}_q$, denote $\mathbf{c}\in \mathbb{F}_q^{kr}$ as the vector $\mathbf{c}=(c^{s}g^{t}(c),~(s,t)\in [0,r-1]\times [0,k-1])^{T}.$
    Then, we have
    \begin{equation}\label{eqa-1-1}
     \mathbf{c}=\sum_{(i,j)\in [k]\times [r]}
     \frac{h_{G\setminus\{g_{i}\}}(g(c))}{h_{G\setminus\{g_{i}\}}(g_{i})}
     \frac{h_{\tilde{A}_i\setminus\{a_{i,j}\}}(c)}{h_{\tilde{A}_i\setminus\{a_{i,j}\}}(a_{i,j})}
     \mathbf{a}_{i,j}.
    \end{equation}
    \item [3)] For every $1\leq i\leq k$, denote $\tilde{\mathbf{A}}_i$ as the $r\times r$ Vandermonde matrix generated by $\tilde{A}_{i}$. Then, we have \begin{equation}\label{eqa-1-2}
    \left(\begin{array}{cccc}
       \tilde{\mathbf{A}}_{1} & \tilde{\mathbf{A}}_{2} & \cdots & \tilde{\mathbf{A}}_{k}\\
       g_1\tilde{\mathbf{A}}_{1} & g_2\tilde{\mathbf{A}}_{2} & \cdots & g_k\tilde{\mathbf{A}}_{k}\\
       \vdots & \vdots &  & \vdots\\
       g_1^{k-2}\tilde{\mathbf{A}}_{1} & g_2^{k-2}\tilde{\mathbf{A}}_{2} & \cdots & g_k^{k-2}\tilde{\mathbf{A}}_{k}
    \end{array}\right)\cdot
    \left(\begin{array}{c}
       \frac{\tilde{\mathbf{A}}_{1}^{-1}}{h_{G\setminus\{g_1\}}(g_1)} \\
       \frac{\tilde{\mathbf{A}}_{2}^{-1}}{h_{G\setminus\{g_2\}}(g_2)}\\
       \vdots \\
       \frac{\tilde{\mathbf{A}}_{k}^{-1}}{h_{G\setminus\{g_k\}}(g_k)} 
    \end{array}\right)=\mathbf{0}_{kr\times r}.
    \end{equation}
\end{itemize}
\end{lemma}

\begin{proof}
There is a one-to-one correspondence between a polynomial $f$ in $\text{Span}_{\mathbb{F}_q}\{x^sg^{t}(x)\}_{0\leq s\leq r-1 \atop 0\leq t\leq k-1}$ and its coefficient vector $\mathbf{v}_f\in \mathbb{F}_q^{kr}$ w.r.t. basis $\{x^sg^{t}(x)\}_{0\leq s\leq r-1 \atop 0\leq t\leq k-1}$. Clearly, for every $c\in \mathbb{F}_q$, $f(c)=\mathbf{v}_f\cdot \mathbf{c}$. Next, we proceed to prove the three results separately.

\begin{itemize}
    \item [1.] Let $\mathbf{A}$ be the $kr\times kr$ matrix with $\mathbf{a}_{i,j}$ as its $((i-1)r+j)$-th column. Next, we prove the first result by showing that $rank(\mathbf{A})=kr$. 
    
    Assume that there is a vector $\mathbf{v}_f\in \mathbb{F}_q^{kr}$ such that $\mathbf{v}_f\cdot\mathbf{A} =\mathbf{0}$. Thus, $f(a_{i,j})=0$ for every $(i,j)\in [k]\times [r]$. For each $1\leq i\leq k$, denote 
    $$f_{i}(x)=\sum_{j=1}^{r}\left(\sum_{s=1}^{k}f_{s,j}g_i^{s-1}\right)x^{j-1}.$$
    Then, we have $f_{i}(a_{i,j})=0$ for every $(i,j)\in [k]\times [r]$. Note that $f_{i}(x)$ is a polynomial of degree at most $r-1$. Thus, by $f_i(a_{i,j})=0$ for every $1\leq j\leq r$, $f_i$ is the zero polynomial. This implies that $\sum_{s=1}^{k}f_{s,j}g_i^{s-1}=0$ for every $(i,j)\in [k]\times [r]$. Denote $\tilde{f}_j(x)=\sum_{s=1}^{k}f_{s,j}x^{s-1}$. Since $g_{i}$'s are pairwise distinct and $\deg(\tilde{f}_j)\leq k-1$, we can obtain $f_{s,j}=0$ for every $(s,j)\in [k]\times [r]$. Therefore, we can conclude that $\mathbf{v}_f=\mathbf{0}$ and $rank(\mathbf{A})=kr$.
    \item [2.] Since $\{\mathbf{a}_{i,j}\}_{1\leq i\leq k \atop 1\leq j\leq r}$ are linearly independent over $\mathbb{F}_q$, $\mathbf{c}$ can be uniquely expressed as a linear combination of $\{\mathbf{a}_{i,j}\}_{1\leq i\leq k \atop 1\leq j\leq r}$. Assume that $\mathbf{c}=\sum_{(i,j)\in [k]\times[r]}\eta_{i,j}\mathbf{a}_{i,j}$. Next, we show that
    \begin{equation}\label{eqa-1-3}
    \eta_{i,j}=\frac{h_{G\setminus\{g_{i}\}}(g(c))}{h_{G\setminus\{g_{i}\}}(g_{i})}
     \frac{h_{\tilde{A}_i\setminus\{a_{i,j}\}}(c)}{h_{\tilde{A}_i\setminus\{a_{i,j}\}}(a_{i,j})}.
    \end{equation}
     
     For any $f\in \text{Span}_{\mathbb{F}_q}\{x^sg^{t}(x)\}_{0\leq s\leq r-1 \atop 0\leq t\leq k-1}$, $f(c)=\mathbf{v}_{f}\cdot \mathbf{c}^{T}$. Thus, $f(c)=\sum_{(i,j)\in [k]\times[r]}\eta_{i,j} f({a}_{i,j})$. For $(i,j)\in [k]\times[r]$, define 
     $$f_{i,j}(x)=h_{G\setminus\{g_{i}\}}(g(x))h_{\tilde{A}_i\setminus\{a_{i,j}\}}(x).$$
     Then,  $f_{i,j}(c)=h_{G\setminus\{g_{i}\}}(g(c))h_{\tilde{A}_i\setminus\{a_{i,j}\}}(c)$ and for $(s,t)\in [k]\times[r]$,
     $$
     f_{i,j}(a_{s,t})=\left\{
         \begin{array}{cc}
         h_{G\setminus\{g_{i}\}}(g_{i})h_{\tilde{A}_i\setminus\{a_{i,j}\}}(a_{i,j}) & \text{if}~(s,t)=(i,j);\\
         0 & \text{otherwise}.
         \end{array}
         \right.
     $$
     Thus, (\ref{eqa-1-3}) follows directly from $f_{i,j}(c)=\sum_{(s,t)\in [k]\times[r]}\eta_{s,t} f_{i,j}({a}_{s,t})$.
     \item [3.] To start with, we consider the toy case when $r=1$. When $r=1$, we have $\tilde{\mathbf{A}}_{i}=(1)$ and $$\left(\begin{array}{ccc}
       \tilde{\mathbf{A}}_{1} &  \cdots & \tilde{\mathbf{A}}_{k}\\
       g_1\tilde{\mathbf{A}}_{1} & \cdots & g_k\tilde{\mathbf{A}}_{k}\\
       \vdots &  & \vdots\\
       g_1^{k-2}\tilde{\mathbf{A}}_{1} & \cdots & g_k^{k-2}\tilde{\mathbf{A}}_{k}
    \end{array}\right)=\left(\begin{array}{ccc}
       1 & \cdots & 1\\
       g_1 & \cdots & g_k\\
       \vdots &  & \vdots\\
       g_1^{k-2} & \cdots & g_k^{k-2}
    \end{array}\right).$$
    Note that the RHS of the above equality is the $(k-1)\times k$ Vandermonde matrix generated by $G$. Therefore, by Cramer's rule, we have 
    $$\left(\begin{array}{ccc}
       1 & \cdots & 1\\
       g_1 & \cdots & g_k\\
       \vdots &  & \vdots\\
       g_1^{k-2} & \cdots & g_k^{k-2}
    \end{array}\right)\cdot
    \left(\begin{array}{c}
       \frac{1}{h_{G\setminus\{g_1\}}(g_1)} \\
       \vdots \\
       \frac{1}{h_{G\setminus\{g_k\}}(g_k)} 
    \end{array}\right)=\mathbf{0}.$$
    This confirms (\ref{eqa-1-2}) when $r=1$.
    
    For $r\geq 1$, the LHS of (\ref{eqa-1-2}) equals to
    $$
    \left(\begin{array}{cccc}
       \mathbf{I} & \mathbf{I} & \cdots & \mathbf{I}\\
       g_1\mathbf{I} & g_2\mathbf{I} & \cdots & g_k\mathbf{I}\\
       \vdots & \vdots &  & \vdots\\
       g_1^{k-2}\mathbf{I} & g_2^{k-2}\mathbf{I} & \cdots & g_k^{k-2}\mathbf{I}
    \end{array}\right)\cdot
    \left(\begin{array}{c}
       \frac{\mathbf{I}}{h_{G\setminus\{g_1\}}(g_1)} \\
       \frac{\mathbf{I}}{h_{G\setminus\{g_2\}}(g_2)}\\
       \vdots \\
       \frac{\mathbf{I}}{h_{G\setminus\{g_k\}}(g_k)} 
    \end{array}\right).
    $$
    Then, (\ref{eqa-1-2}) follows from the case when $r=1$.
\end{itemize}
\end{proof}

\end{document}